\documentclass[12pt]{article}

\usepackage{lipsum}
\usepackage{amsthm,amssymb,amsmath}

\usepackage{xr-hyper}
\usepackage{amsthm}
\usepackage{fullpage}
\usepackage{footnote}
\usepackage{tabu}
\usepackage{booktabs}
\usepackage[title]{appendix}
\usepackage{scribe}
\usepackage{listings}
\usepackage[mathscr]{euscript}

\usepackage{verbatim}
\usepackage{rotating}
\usepackage{float}
\usepackage{setspace}
\usepackage{url}
\usepackage{enumitem}
\usepackage{pdfpages}
\usepackage{blkarray}
\usepackage{mathtools}

\usepackage{ stmaryrd }
 
\usepackage{graphicx}
\usepackage{subfigure}
\usepackage{multirow}
\usepackage{amsmath}
\usepackage{xcolor}

\providecommand{\keywords}[1]
{
  \small	
  \textbf{\textit{Keywords---}} #1
}
\usepackage{etoolbox}

\makeatletter
\patchcmd{\@makecaption}
  {\parbox}
  {\advance\@tempdima-\fontdimen2} 
  {}{}
\makeatother 

\usepackage{titling} 

\usepackage{amsmath,bbm,amssymb, amsfonts}
\usepackage{mathtools}
\usepackage{graphicx,psfrag,epsf}
\usepackage{enumitem}
\usepackage{natbib}
\usepackage{url} 
\usepackage{color}

\usepackage{algorithm}
\usepackage{algpseudocode}

\algnewcommand\INPUT{\item[\textbf{Input:}]}%
\algnewcommand\OUTPUT{\item[\textbf{Output:}]}%
\usepackage{booktabs}
\usepackage{siunitx}
\usepackage{array}
\usepackage{rotating}

\usepackage{adjustbox}

\usepackage{tikz}

\usepackage{url}

\usepackage{multicol}
\usepackage{xcolor}

\makeatletter
\def\breve{\mathpalette\wide@breve}
\def\wide@breve#1#2{\sbox\z@{$#1#2$}%
	\mathop{\vbox{\m@th\ialign{##\crcr
				\kern0.08em\brevefill#1{0.8\wd\z@}\crcr\noalign{\nointerlineskip}%
				$\hss#1#2\hss$\crcr}}}\limits}
\def\brevefill#1#2{$\m@th\sbox\tw@{$#1($}%
	\hss\resizebox{#2}{\wd\tw@}{\rotatebox[origin=c]{90}{\upshape(}}\hss$}
\makeatletter

\def\gobblestop#1#2{#1}
\def\killstop{%
	\aftergroup\gobblestop
}

\def\bx{{\bf x}}

\def\bA{{\bf A}}
\def\bB{{\bf B}}
\def\bC{{\bf C}}
\def\bD{{\bf D}}
\def\bE{{\bf E}}

\def\bI{{\bf I}}

\def\bL{{\bf L}}

\def\bR{{\bf R}}
\def\bS{{\bf S}}

\def\bU{{\bf U}}
\def\bV{{\bf V}}

\def\bX{{\bf X}}
\def\bY{{\bf Y}}

\def\thick#1{\hbox{\rlap{$#1$}\kern0.25pt\rlap{$#1$}\kern0.25pt$#1$}}

\def\smbalpha{\boldsymbol{{\scriptstyle{\alpha}}}}

\def\smbalpha{\widehat{\smbalpha}}

\def\hbar{\bar{ h}}

\def\mybox#1{\vskip1mm \begin{center}
        \hspace{.0\textwidth}\vbox{\hrule\hbox{\vrule\kern6pt
\parbox{.9\textwidth}{\kern6pt#1\vskip6pt}\kern6pt\vrule}\hrule}
        \end{center} \vskip-5mm}
\def\lboxit#1{\vbox{\hrule\hbox{\vrule\kern6pt
      \vbox{\kern6pt#1\vskip6pt}\kern6pt\vrule}\hrule}}

\def\thickboxit#1{\vbox{{\hrule height 1mm}\hbox{{\vrule width 1mm}\kern6pt
          \vbox{\kern6pt#1\kern6pt}\kern6pt{\vrule width 1mm}}
               {\hrule height 1mm}}}

\def\fat#1{\hbox{\rlap{$#1$}\kern0.25pt\rlap{$#1$}\kern0.25pt$#1$}}

\newcolumntype{R}{@{\extracolsep{0.5cm}}r@{\extracolsep{0pt}}}%
\newcolumntype{E}{@{\extracolsep{0.25cm}}c@{\extracolsep{0pt}}}%

\makeatletter
\newcommand{\distas}[1]{\mathbin{\overset{#1}{\kern\z@\sim}}}%
\makeatother

\usepackage{xr-hyper}
\usepackage[colorlinks, linkcolor=blue, anchorcolor=blue, citecolor=blue]{hyperref}
\usepackage[capitalize]{cleveref}

\newtheorem{thm}{Theorem}[section]

\def\bSddONE{{\bS_{\boldsymbol{\cdot}\boldsymbol{\cdot}}^{(1)}}}

\def\bSddl{{\bS_{\boldsymbol{\cdot}\boldsymbol{\cdot}}^{(l)}}}
\def\bSddlt{{\bS_{\boldsymbol{\cdot}\boldsymbol{\cdot}}^{(l')}}}
\def\bSddik{{\bS_{\boldsymbol{\cdot}\boldsymbol{\cdot}}^{(i-K)}}}
\def\bSdds{{\bS_{\boldsymbol{\cdot}\boldsymbol{\cdot}}^{(s)}}}

\def\bSddL{{\bS_{\boldsymbol{\cdot}\boldsymbol{\cdot}}^{(L)}}}
\def\bSddONEbar{{\bar{\bS}_{\boldsymbol{\cdot}\boldsymbol{\cdot}}^{(1)}}}
\def\bSddlbar{{\bar{\bS}_{\boldsymbol{\cdot}\boldsymbol{\cdot}}^{(l)}}}
\def\bSddLbar{{\bar{\bS}_{\boldsymbol{\cdot}\boldsymbol{\cdot}}^{(L)}}}

\def\bY{\mathbf Y}
\def\bYd{\mathbf{Y}_{\boldsymbol{\cdot}}}
\def\bYhdk{\mathbf{Y}_{h\boldsymbol{\cdot}}^{(k)}}
\def\bYhdkT{\mathbf{Y}_{h\boldsymbol{\cdot}}^{(k)T}}
\def\bYdd{\mathbf{Y}_{\boldsymbol{\cdot\cdot}}}
\def\bX{\mathbf X}
\def\bx{{\bf x}}

\def\bL{{\bf L}}
\def\bC{{\bf C}}

\def\bA{\mathbf A}
\def\bB{\mathbf B}
\def\bBdk{\mathbf{B}_{\boldsymbol{\cdot}}^{(k)}}
\def\bBdhk{\mathbf{B}_{\boldsymbol{\cdot}h}^{(k)}}
\def\bBd{\mathbf{B}_{\boldsymbol{\cdot}}}
\def\bU{\mathbf U}

\def\bX{\mathbf X}
\def\bXdd{\mathbf{X}_{\boldsymbol{\cdot \cdot}}}

\def\bXddl{\mathbf{X}_{\boldsymbol{\cdot \cdot}}^{(l)}}
\def\bx{\mathbf x}
\def\bI{\mathbf I}

\def\bR{\mathbf R}
\def\bS{\mathbf S}

\def\bV{{\bf V}}

\def\bD{\mathbf D}
\def\bE{\mathbf E}
\def\bEdd{\mathbf{E}_{\boldsymbol{\cdot}\boldsymbol{\cdot}}}

\usepackage[left=0.9in,right=0.9in,top=0.95in,bottom=0.95in]{geometry}

\makeatletter
\newcommand*{\addFileDependency}[1]{
  \typeout{(#1)}
  \@addtofilelist{#1}
  \IfFileExists{#1}{}{\typeout{No file #1.}}
}
\makeatother
 
\newcommand*{\myexternaldocument}[1]{%
    \externaldocument{#1}%
    \addFileDependency{#1.tex}%
    \addFileDependency{#1.aux}%
}

\myexternaldocument{supplement}

\newcommand{\blind}{0}

\begin{document}



 

\if0\blind
{
	\title{\bf BARRR: Bidimensional Augmented Reduced Rank Regression for Multi-View Multi-Cohort Data}
	
	\author{
        Jiuzhou Wang$^{1}$\thanks{Corresponding author. Mailing address: University Office Plaza, Ste 200. 2221 University Ave SE, Minneapolis, MN 55414. Phone: +1 612-625-0651. Email: jiuzhou.wang.jz@gmail.com.}, 
        Eric F. Lock$^1$ \\ \\
	    $^1$Division of Biostatistics and Health Data Science, School of Public Health,\\ 
        University of Minnesota 
    }

	\date{}
	\maketitle
} \fi


\begin{abstract}
Bidimensional data integration across multiple cohorts (i.e., sample groups) and views (i.e. sources or modalities) is crucial in biomedical research to derive holistic insights from heterogeneous datasets. However, most existing integrative methods focus on a single dimension, such as multi-view integration across distinct sources or multi-cohort integration across distinct sample groups. Existing approaches for bidimensional integration are scarce and predominantly unsupervised, overlooking the valuable insights offered by covariates. To address this gap, we introduce the Bidimensional Augmented Reduced Rank Regression (baRRR) method. Our approach models covariate-related and unrelated effects concurrently, while also decomposing effect levels—global, partially-shared, and individual—across both dimensions. Extensive simulations demonstrate the advantages of our method with respect to signal recovery, outcome prediction, and missing data imputation. We apply our approach to data from The Cancer Genome Atlas (TCGA), using somatic mutations as covariates to inform patterns of variability across different omics platforms and cancer types. These studies highlight its potential as a robust and adaptable tool for integrative analysis in biomedical research.  \\ \\ 
\end{abstract}%

\def\spacingset#1{\renewcommand{\baselinestretch}%
{#1}\small\normalsize} \spacingset{1.5}

\newpage
\spacingset{1.725} 
\setlength{\abovedisplayskip}{7pt}%
\setlength{\belowdisplayskip}{7pt}%
\setlength{\abovedisplayshortskip}{5pt}%
\setlength{\belowdisplayshortskip}{5pt}%

%


\section{Introduction}
\label{paper3_sec:intro}

The rapid progression of high-throughput technologies in biomedicine has demanded new methodology for integrating multiple large datasets.  Increasingly, studies require the integration of so-called bidimensional data, comprising both multi-view (also called multi-source or multi-modality) and multi-cohort data. The former captures measurements from varied sources for the same sample set, while the latter captures measurements from the same features across disparate sample sets. 
A quintessential example of this multi-faceted data structure is the Cancer Genome Atlas (TCGA) Pan-Cancer Project \citep{hoadley2018cell}. From TCGA we consider data from four distinct yet biologically interconnected molecular sources—mRNA, miRNA, DNA methylation, and  proteins—spanning 33 tumor sample cohorts, such as those for brain, breast, and lung cancers. 
We also consider data for somatic mutations, which are the primary underlying cause of cancer and can affect molecular function at the other molecular levels \citep{vogelstein2013cancer}.  
With these data, we seek to answer the following questions: 
\begin{enumerate}
 \item What molecular signatures are indicative of each cancer type? 
 \item What effect do somatic mutations have on downstream molecular activity, and are these effects shared across cancer types?
 \item What are the primary patterns of molecular heterogeneity, and are these patterns shared across cancer types? 
\end{enumerate}
To address these questions comprehensively, we need to use the full collection of data from TCGA with new methodology to properly decompose the different molecular relationships across different sources and cohorts.  

Existing integrative analyses predominantly focus on one dimension, i.e.
multi-view or multi-cohort (see \citet{li2026integrative} for a recent
review). Several methods are unsupervised and perform low-rank matrix decompositions analogous to that in principal component analysis (PCA). Joint and individual variation explained (JIVE, \citep{Lock2013JIVE}) is one of the earliest methods to integrate data from multiple sources by decomposing the signals into joint and individual effects. Several extensions have been proposed to decompose the signals with low-rank approximations \citep{o2019linked,Irina2019SLIDE,zhu2020generalized, kim2023bayesian,yi2023hierarchical}, while some of them use the decomposed signals to predict outcomes \citep{kaplan2017prediction, palzer2022sjive,li2022integrative,samorodnitsky2022bayesian}. Beyond unsupervised methods, supervised multi-source factorization methods have been proposed that only consider the covariate-related effects to predict outcomes \citep{Chen2013ANN,she2013reduced,Li2019iRRR};  these can be considered as extensions of reduced-rank regression, as they assume a low-rank decomposition on the coefficients for a predictive model.

The literature on multi-cohort data integration remains relatively limited. JICO \citep{wang2022joint} proposes a JIVE-like decomposition for covariates when predicting only one univariate outcome. Multiple augmented reduced rank regression (maRRR) \citep{wang2023multiple} decomposes the multi-variate outcomes into covariate effects and auxiliary effects while allowing joint, partial and individual signals, which is suitable for multi-cohort data integration with covariates but not for bidimensional data.

Only a few methods address the decomposition of bidimensional data. BIDIFAC \citep{park2020integrative} extends JIVE into a bidimensional setting, feature global, row-shared, column-shared and individual effects. BIDIFAC+ \citep{Lock2022BADIFACplus} extends BIDIFAC further by allowing partially shared effects with any combinations of views and cohorts. An empirical Bayes formulation \citep{lock2024empirical} estimates the same kind of bidimensional decomposition while determining the amount of shrinkage for each component from the data. Taking shared signals into consideration via the specification of priors, a Bayesian data integration method \citep{argelaguet2020mofa+} only models individual effects without shared structures to decompose. Nevertheless, these methods are entirely unsupervised, often overlooking the invaluable information presented by covariates. Of particular relevance, there is a method using decomposed multi-view multi-cohort covariate information to predict low-dimensional outcomes \citep{butts2021accounting}; however, it ignores shared structures between different views and cohorts.

We propose a general approach that builds on the maRRR method in several ways. First,  bidimensionally linked data are decomposed into shared and unshared structures across both row and column sets. Second, we allow for multiple covariate sets to further decompose the signals into their covariate-related and covariate-unrelated parts.  Third, we accommodate cohort specific intercepts within our framework, allowing for the decomposition of mean effects and variation about the mean. Fourth, we present new theoretical results on convergence of the iterative algorithms. 

Low-rank imputation methods such as SVDImpute \citep{troyanskaya2001missing}, softImpute \citep{mazumder2010spectral} and convex matrix completion \citep{candes2009exact} apply to a single matrix, while linked-matrix approaches \citep{o2019linked,lock2024empirical} can impute blockwise missingness across matrices; none of these use covariates. Our approach uses covariates to impute various forms of missing data, including ``blockwise'' missingness in which entire rows or columns are missing (e.g., entire data sources missing for a sample). For our pan-cancer application, this allows all available data to inform an interpretable decomposition of variation into (1) molecular signatures for each cancer type, (2) variation driven by somatic mutations across molecular levels and cancer types, and (3) patterns of molecular heterogeneity across cancer types that are unrelated to somatic mutations.


\section{Proposed Model}
\label{paper3_sec:method}

In the bidimensional setting, $\bX_{ij}: p_i \times n_j$ symbolizes a data matrix for the $j$'th sample cohort from the $i$'th data source, coupled with covariates $\bY_{hj}: q_h \times n_j$ for the $j$'th sample cohort from the $h$'th data view, where $h=1,...,H$, $i=1,...,I$ and $j=1,...,J$. Concatenations across the sources and cohorts are denoted by subscripts $\boldsymbol{\cdot}$ in the first and second 
 position respectively, e.g., $\bX_{i\boldsymbol{\cdot}} = [\bX_{i1},\bX_{i2},...,\bX_{iJ}] : p_i \times n$ representing outcomes from $i$th source and  $\bY_{h\boldsymbol{\cdot}} = [\bY_{h1},\bY_{h2},...,\bY_{hJ}]: q_h \times n$ representing covariates from $h$th covariate view, as $n = \sum_{j=1}^J n_j$. Therefore, the concatenation of data across all sources and cohorts is denoted by
\begin{align*}
\bX_{\boldsymbol{\cdot}\boldsymbol{\cdot}} = \begin{bmatrix}
\bX_{11} & \bX_{12} & \cdots & \bX_{1J}\\
\vdots & \vdots &\cdots &\vdots \\
\bX_{I1} & \bX_{I2} & \cdots & \bX_{IJ}
\end{bmatrix}: p \times n, p = \sum_{i=1}^I p_i;\\ \bY_{\boldsymbol{\cdot}\boldsymbol{\cdot}} = \begin{bmatrix}
\bY_{11} & \bY_{12} & \cdots & \bY_{1J}\\
\vdots & \vdots &\cdots &\vdots \\
\bY_{H1} & \bY_{H2} & \cdots & \bY_{HJ}
\end{bmatrix}: q \times n, q = \sum_{h=1}^H q_h.
\end{align*}

For our application, we consider TCGA data $\bXdd$ for several patients across $J=27$ cancer types and $I=4$ omics sources, from which we aim to extract the main signals. These omics sources include gene expression, miRNA, methylation CpG probes, and proteins. We consider $H=2$ covariate views for $\bYdd$, representing (1) somatic mutations for commonly mutated genes and (2) an indicator matrix for the different cancer types. Of note, only $\bXdd$ and $\bYdd$ are observed in the model. We have two primary aims: (1) uncovering multiple shared structures across tumors and omics, regardless of their relationship to somatic mutations; and (2) identifying somatic mutations that are highly related to specific shared structures. 
This involves estimating low-rank coefficient matrices $\bB_{ih}^{(k)}: p_i \times q_h, k = 1,\dots,K_h$ and auxiliary variation structures $\bS_{ij}^{(l)}: p_i \times n_j, l = 1,...,L$. Allowing noise $\bE_{ij}: p_i \times n_j$ in every block, the full model is
\begin{align}\label{paper3_newmodel}
    \bXdd &= \sum^{H}_{h=1} \sum^{K_h}_{k=1} {\mathbf{B}_{\boldsymbol{\cdot}h}^{(k)}} \bY_{h\boldsymbol{\cdot}}^{(k)} + \sum^{L}_{l=1}  \bSddl + \bEdd 
    \end{align}
    where $\bY_{h\boldsymbol{\cdot}}^{(k)} = [\bY_{h1}^{(k)} , \bY_{h2}^{(k)} , \cdots , \bY_{hJ}^{(k)}]: q_h \times n, \bSddl : p \times n, \bEdd : p \times n, $
\begin{align*}
{\mathbf{B}_{\boldsymbol{\cdot}h}^{(k)}} = \begin{bmatrix}
\bB_{1h}^{(k)} \\
\vdots \\
\bB_{Ih}^{(k)} 
\end{bmatrix}: p \times q_h,
\bSddl = \begin{bmatrix}
\bS_{11}^{(l)} & \bS_{12}^{(l)} & \cdots & \bS_{1J}^{(l)}\\
\vdots & \vdots &\cdots &\vdots \\
\bS_{I1}^{(l)} & \bS_{I2}^{(l)} & \cdots & \bS_{IJ}^{(l)}
\end{bmatrix},
\bEdd = \begin{bmatrix}
\bE_{11} & \bE_{12} & \cdots & \bE_{1J}\\
\vdots & \vdots &\cdots &\vdots \\
\bE_{I1} & \bE_{I2} & \cdots & \bE_{IJ}
\end{bmatrix}.
\end{align*}        
Which blocks of $\bB_{ih}^{(k)}$, $\bY_{hj}^{(k)}$ and $\bS_{ij}^{(l)}$ are nonzero is governed by binary indicator matrices $\bR_{Bh}: I\times K_h$, $\bC_{Yh}: J\times K_h$ for $h=1,\dots,H$, and $\bR_S: I\times L$, $\bC_S: J\times L$: 
\begin{align*}
&  \bB_{ih}^{(k)}= 
  \begin{dcases} 
      \boldsymbol{0}_{p_i\times q_h} & \textit{if }  \bR_{Bh}[i,k] = 0 \\ 
  \bU_{Bh,i}^{(k)}\bV_{Bh}^{(k)T} & \textit{if }   \bR_{Bh}[i,k] =  1,
  \end{dcases} 
    \bY_{hj}^{(k)}= 
  \begin{dcases} 
      \boldsymbol{0}_{q_h\times n_j} & \textit{if }  \bC_{Yh}[j,k] = 0 \\ 
  \bY_{hj} & \textit{if }   \bC_{Yh}[j,k] =  1,
  \end{dcases} 
  \\
 &   \bS_{ij}^{(l)} = 
  \begin{dcases} 
  \boldsymbol{0}_{p_i\times n_j} & \textit{if }   \bC_S[j,l] = 0 \textit{ or } \bR_S[i,l] = 0 \\ 
  \bU_{S,i}^{(l)}\bV_{S,j}^{(l)T} & \textit{if }   \bC_S[j,l] = \bR_S[i,l] = 1.
  \end{dcases}
\end{align*}
$\bU_{Bh,i}^{(k)}$,$\bU_{S,i}^{(l)}$ and $\bV_{S,j}^{(l)T}$ can be concatenated, i.e. $\bU_{Bh}^{(k)} = \begin{bmatrix}
\bU_{Bh,1}^{(k)} \\
\vdots \\
\bU_{Bh,I}^{(k)} 
\end{bmatrix}$, $\bU_{S}^{(l)} = \begin{bmatrix}
\bU_{S,1}^{(l)} \\
\vdots \\
\bU_{S,I}^{(l)} 
\end{bmatrix}$ and $\bV_{S}^{(l)} = [\bV_{S,1}^{(l)},...,\bV_{S,J}^{(l)}]$ would be the loading matrices and score matrix for $\bBdhk$ and $\bSddl$ if we write $\bBdhk = \bU_{Bh}^{(k)}\bV_{Bh}^{(k)T}$ and $\bSddl = \bU_{S}^{(l)}\bV_{S}^{(l)T}$. 



A module is the basic unit of our decomposition. In the one-dimensional setting a module is pinned to a subset of cohorts alone; here each module occupies the blocks of $\bXdd$ selected by a pair of subsets, one over the $I$ sources and one over the $J$ cohorts. For an auxiliary module $\bSddl$, the column $\bR_S[:,l]$ selects the sources and $\bC_S[:,l]$ selects the cohorts; outside those blocks the module is identically zero. The covariate-driven modules $\bBdhk \bY_{h\boldsymbol{\cdot}}^{(k)}$ are indexed the same way through $\bR_{Bh}[:,k]$ and $\bC_{Yh}[:,k]$, with the difference that each is tied to a single covariate view $h$ and carries a coefficient block of width $q_h$. Separate indicator matrices and a separate module count $K_h$ are therefore specified per view, so covariate sets with distinct sharing patterns need not be forced into a common structure. Within its blocks every module is constrained to be low rank, through the factorizations $\bSddl = \bU_S^{(l)}\bV_S^{(l)T}$ and $\bBdhk = \bU_{Bh}^{(k)}\bV_{Bh}^{(k)T}$ introduced above. The two families play complementary roles: $\bBdhk \bY_{h\boldsymbol{\cdot}}^{(k)}$ captures structure attributed to the observed covariates, while $\bSddl$ captures structured variation not attributed to them.

\cref{paper3_fig:model} illustrates the construction on a small example.

\begin{figure}[H]
        \centering
\includegraphics[height=16cm,width=13.5cm]{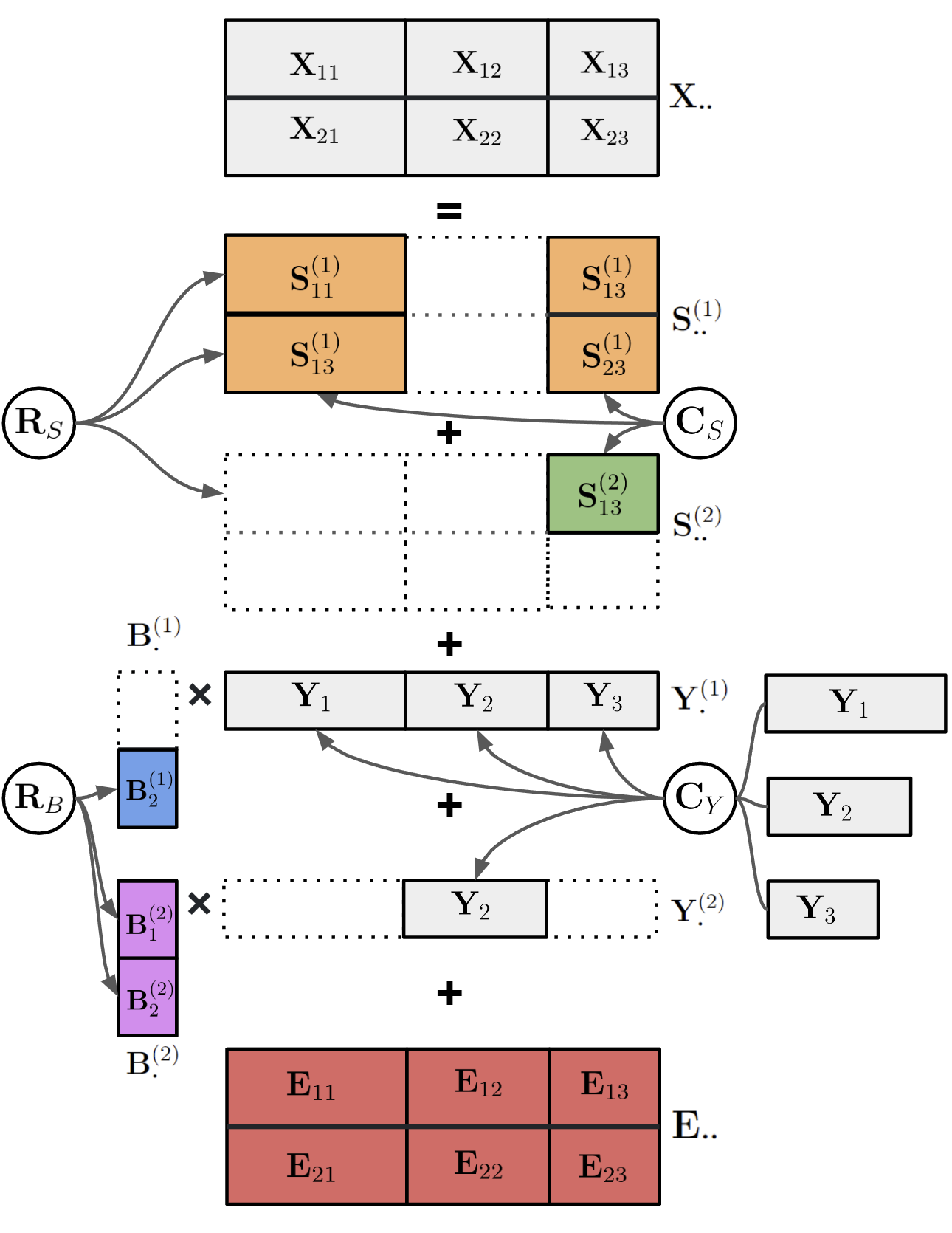}
\caption{
Schematic of baRRR with $I=2$ sources and $J=3$ cohorts, with $H=1$ view of covariates: $\bXdd = \bBd^{(1)} \bYd^{(1)} + \bBd^{(2)} \bYd^{(2)} + \bSddONE + \bS_{\boldsymbol{\cdot\cdot}}^{(2)} + \bEdd$. Grey blocks are observed, namely
$\bXdd = \begin{bmatrix}
\bX_{11} & \bX_{12} & \bX_{13} \\
\bX_{21} & \bX_{22} & \bX_{23}
\end{bmatrix}$ and $\bYd = [\bY_1, \bY_2, \bY_3]$; colored blocks are estimated. Fixing the
four indicator matrices $\bR_S = \begin{bmatrix} 1 & 1 \\ 1 & 0 \end{bmatrix}$,
$\bC_S^T = \begin{bmatrix} 1 & 0 & 1 \\ 0 & 0 & 1 \end{bmatrix}$,
$\bR_B = \begin{bmatrix} 0 & 1 \\ 1 & 1 \end{bmatrix}$,
$\bC_Y^T = \begin{bmatrix} 1 & 1 & 1 \\ 0 & 1 & 0 \end{bmatrix}$ pins each module to its blocks of $\bXdd$: $\bSddONE$ spans both sources but only cohorts 1 and 3, $\bS_{\boldsymbol{\cdot\cdot}}^{(2)}$ only source 1 of cohort 3, $\bBd^{(1)} = [\mathbf{0}, \bB_2^{(1)}]^T$ acts on source 2 across all cohorts, and $\bBd^{(2)} = [\bB_1^{(2)}, \bB_2^{(2)}]^T$ on both sources within cohort 2. Hence $\bYd^{(1)} = [\bY_1, \bY_2, \bY_3]$, $\bYd^{(2)} = [\mathbf{0}, \bY_2, \mathbf{0}]$,
$\bSddONE = \begin{bmatrix}
\bU_{S1}^{(1)} \bV_{S1}^{(1)T} & \mathbf{0} & \bU_{S1}^{(1)} \bV_{S3}^{(1)T} \\
\bU_{S2}^{(1)} \bV_{S1}^{(1)T} & \mathbf{0} & \bU_{S2}^{(1)} \bV_{S3}^{(1)T}
\end{bmatrix}$,
$\bS_{\boldsymbol{\cdot\cdot}}^{(2)} = \begin{bmatrix}
\mathbf{0} & \mathbf{0} & \bU_{S1}^{(2)} \bV_{S3}^{(2)T} \\
\mathbf{0} & \mathbf{0} & \mathbf{0}
\end{bmatrix}$,
and $\bEdd = \begin{bmatrix}
\bE_{11} & \bE_{12} & \bE_{13} \\
\bE_{21} & \bE_{22} & \bE_{23}
\end{bmatrix}$.
}

\label{paper3_fig:model}
\end{figure}

\section{Objective Function} \label{paper3_sec:obejective}

Our estimator for model~\eqref{paper3_newmodel} balances two goals: reproducing $\bXdd$ closely, as measured by squared error, and keeping each module low rank, as enforced by a nuclear norm penalty with its own weight:
\begin{small}
\begin{align}\label{paper3_ObjZ}
    \min_{\{\bBdhk\}^{K_h}_{k=1},\{\bS_{\boldsymbol{\cdot}\boldsymbol{\cdot}}^{(l)}\}^L_{l=1}} \{ \frac{1}{2} ||\bX_{\boldsymbol{\cdot}\boldsymbol{\cdot}} - \sum^{H}_{h=1} \sum^{K_h}_{k=1} \bBdhk \bYhdk - \sum^{L}_{l=1}  \bS_{\boldsymbol{\cdot}\boldsymbol{\cdot}}^{(l)}||_F^2 + \sum^{H}_{h=1} \sum^{K_h}_{k=1} \lambda_{Bh}^{(k)} ||\bBdhk||_{*} +\sum^{L}_{l=1} \lambda_S^{(l)} ||\bS_{\boldsymbol{\cdot}\boldsymbol{\cdot}}^{(l)}||_{*}  \}
\end{align}    
\end{small}
Each $||\cdot||_*$ is a nuclear norm, the sum of a matrix's singular values. Because it is convex, \eqref{paper3_ObjZ} is a convex problem; its effect on each module is to shrink singular values toward zero and set the smallest exactly to zero, so module ranks are chosen by the fit rather than fixed in advance.

Two previously proposed methods arise as special cases of \eqref{paper3_newmodel}:
\begin{enumerate}
    \item maRRR \citep{wang2023multiple}, for multi-cohort data from a
    single source with a single covariate view, obtained with $I=H=1$.
    \item BIDIFAC+ \citep{Lock2022BADIFACplus}, the fully unsupervised
    decomposition, obtained by removing the covariate terms ($H=0$).
\end{enumerate}

\cite{mazumder2010spectral} and others have noted that the nuclear-norm penalty is equivalent to applying an additive $L_2$ penalty to the factors in a low-rank decomposition. This equivalence yields an alternative expression for our objective in \eqref{paper3_ObjZ}:
\begin{align}\label{paper3_ObjUV}
    &\min_{\substack{\{\bU_{Bh}^{(k)},\bV_{Bh}^{(k)}\}^{K_h}_{k=1}, \{\bU_S^{(l)},\bV_S^{(l)}\}^L_{l=1}}}  \frac{1}{2}\{ ||\bXdd - \sum^H_{h=1} \sum^{K_h}_{k=1} \bU_{Bh}^{(k)}\bV_{Bh}^{(k)T} \bYhdk - \sum^{L}_{l=1}  \bU_S^{(l)}\bV_S^{(l)T}||_F^2 + \nonumber\\
    & \sum^H_{h=1} \sum^{K_h}_{k=1} \lambda_{Bh}^{(k)} (||\bU_{Bh}^{(k)}||_{F}^2 +||\bV_{Bh}^{(k)}||_{F}^2) +\sum^{L}_{l=1} \lambda_S^{(l)}( ||\bU_S^{(l)}||_{F}^2 + ||\bV_S^{(l)}||_{F}^2).  \}
\end{align}
Here, we only need to set the total number of columns for each of $\bU_{Bh}^{(k)}$, $\bV_{Bh}^{(k)}$, $\bU_S^{(l)}$, and $\bV_S^{(l)}$.  This gives an upper bound on the rank of each $\bBdhk$ and $\bSddl$, i.e. $r_{Bh,upper}$ and $r_{S,upper}$.    
However, the actual ranks of the solution may be lower because the nuclear norm penalty promotes rank sparsity. If the upper bounds are sufficiently large, the solution will coincide with that in ~\eqref{paper3_ObjZ}. The following theorem, establishing this equivalence, is a natural extension to Theorem 1 in \cite{wang2023multiple}.

\begin{thm}\label{paper3_equ}
The solutions to objective functions (\ref{paper3_ObjZ}) and (\ref{paper3_ObjUV}) are equivalent, with $\bS_{\boldsymbol{\cdot}\boldsymbol{\cdot}}^{(l)} = \bU_S^{(l)}\bV_S^{(l)T}$ for $l=1,\hdots,L$ and $\bBdhk=\bU_{Bh}^{(k)}\bV_{Bh}^{(k)T}$ for $h=1,\hdots,H$ and $k=1,\hdots,K$. 
\end{thm}

In Section~\ref{paper3_sec:est} we demonstrate the optimization algorithm, prove the convergence of our parameter estimations and describe a random matrix theory approach to automatically select the nuclear norm penalty for the different modules.

\section{Estimation}
\label{paper3_sec:est}

\subsection{Optimization algorithms}
\label{paper3_subsec:optimization}
Both formulations are optimized by cycling through the modules, holding all but one fixed and solving the resulting subproblem exactly. Under \eqref{paper3_ObjUV} a block is a single factor matrix $\bU_{Bh}^{(k)}$, $\bV_{Bh}^{(k)}$, $\bU_S^{(l)}$ or $\bV_S^{(l)}$; since the penalty on the factors is quadratic, each update is a ridge-type least squares problem with a closed-form solution (\cref{ALSUV_barrr}). 
Under \eqref{paper3_ObjZ} a block is an entire module, and when $\bYd^{(k)}$ is semi-orthogonal its update is a single soft-thresholded singular value decomposition of the current residual (\cref{ALSBS_barrr}).

Each epoch is one full pass through all blocks, and both algorithms stop once successive epochs change the estimates by less than a tolerance $\epsilon$:$\sum^{H}_{h=1}\sum^{K_h}_{k=1} ||\widehat{\bB}_{\boldsymbol{\cdot h}}^{(k)} - \widetilde{\bB}_{\boldsymbol{\cdot h}}^{(k)} ||_F^2 + \sum^{L}_{l=1} || \widehat{\bS}_{\boldsymbol{\cdot\cdot}}^{(l)} - \widetilde{\bS}_{\boldsymbol{\cdot\cdot}}^{(l)}||_F^2 < \epsilon$, with hats marking the current epoch and tildes the previous one. A stopping rule based on the change in the objective value works equally well. Both algorithms are blockwise coordinate descent methods, and Section~\ref{paper3_subsec:convergence} establishes their convergence.

For practical analyses, it is advisable to first centralize and standardize each individual outcome matrix $X_{ij}$. This step helps negate the impact of any intercept effect and ensure the standard deviation of noises smaller than 1. Moreover, orthogonalizing $\bYd^{(k)}$ is recommended. After these preliminary steps, penalties can be designated for each term as described in Section~\ref{paper3_subsec:penalty}.

\subsection{Convergence theory}
\label{paper3_subsec:convergence}

Let $\{\bx^k\}$ denote the sequence of solutions at each epoch in \cref{ALSUV_barrr} or \cref{ALSBS_barrr}. Since the loss is monotone decreasing over the epochs and bounded, the sequence of loss $\{F(\bx^k)\}$ has a finite limit point $F(\bar{\bx})$ with enough epochs. In what follows, we establish the convergence of our final estimate; 
this novel result implies convergence of the algorithms in \cite{wang2023multiple} and \cite{Lock2022BADIFACplus} as special cases. The proof of this theorem is in \cref{paper3_proof_thm2}.
\begin{thm}
    Let $\{\bx^k\}$ denote the sequence of solutions at each epoch in \cref{ALSUV_barrr} or \cref{ALSBS_barrr}. The sequence $\{\bx^k\}$ converges to $\bar{\bx}$, which is a critical point of \cref{paper3_ObjUV} (or \cref{paper3_ObjZ}). 
\end{thm}

\subsection{Tuning parameter selection}
\label{paper3_subsec:penalty}

Our penalty weights follow from two random matrix results, Propositions 3 and 4 of \cite{wang2023multiple}, which describe how noise and low-rank signal each shape the singular values of an observed matrix. Because each update in our algorithms changes one module with the rest held fixed, it suffices to choose $\lambda$ for the two single-module problems $\bX = \bS + \bE$ and $\bX = \bB\bY + \bE$.

For $\bX_{m\times n}=\bS_{m\times n}+\bE_{m\times n}$, Proposition 3 bounds the expected largest singular value of pure noise by $\sigma(\sqrt{m}+\sqrt{n})$, while Proposition 4 shows that signal components strong enough to be detectable produce singular values above that level. Thresholding at $\lambda_S=\sigma(\sqrt{m}+\sqrt{n})$ therefore discards components attributable to noise alone while keeping detectable signal. For $\bX_{m\times n}=\bB_{m\times q}\bY_{q\times n}+\bE_{m\times n}$ with $\bY\bY^T=\bI$, right-multiplying by $\bY^T$ gives $\bX\bY^T=\bB+\bE\bY^T$, where $\bE\bY^T$ is again an $m\times q$ matrix of independent Gaussian noise; the same argument yields $\lambda_B=\sigma(\sqrt{m}+\sqrt{q})$.

We apply these choices module by module, taking $m$, $n$ and $q$ to be the dimensions of the blocks each module occupies, and setting $\sigma=1$. This is conservative, since standardizing each outcome matrix to unit standard deviation leaves the noise standard deviation below one. With each $\bY_{h\boldsymbol{\cdot}}^{(k)}$ orthogonalized, the weights are $\lambda_{Bh}^{(k)} = \sqrt{\sum_{i=1}^I p_i\bR_{Bh}[i,k]}+\sqrt{q_h}$ for each $\bBdhk$, and $\lambda_S^{(l)} = \sqrt{\sum_{i=1}^I p_i\bR_S[i,l]}+\sqrt{\sum_{j=1}^J n_j\bC_S[j,l]}$ for each $\bSddl$.

\section{Simulations}
\label{paper3_sec:sim}
We illustrate signal recovery and missing data imputation in a setting with two outcome views, two cohorts, and one covariate view. The model is simplified to $\bXdd = \sum^{K}_{k=1} \bBdk \bY_{\boldsymbol{\cdot}}^{(k)} + \sum^{L}_{l=1}  \bSddl + \bEdd $ where $\bY_{\boldsymbol{\cdot}}^{(k)} = [\bY_{1}^{(k)}  \bY_{2}^{(k)}  \cdots  \bY_{J}^{(k)}]: q \times n$. We consider outcomes $\bXdd$ from two views $p_1 = 150, p_2 = 100$ and two cohorts $n_1 = 200, n_2 = 100$, and covariates $\bY_{\boldsymbol{\cdot}}$ from the same cohorts with $q = 50$. In total 9 covariate-related modules and 9 augmented structures are constructed, with corresponding indicator matrices: 
\\$\bC_Y =\bC_S = \begin{bmatrix}
    1 & 1 & 1 & 1 & 0 & 1 & 1 & 0 & 0\\
    1 & 1 & 1 & 0 & 1 & 0 & 0 & 1 & 1\\
\end{bmatrix}$, 
$\bR_B = \bR_S = \begin{bmatrix}
    1 & 1 & 0 & 1 & 1 & 1 & 0 & 1 & 0\\
    1 & 0 & 1 & 1 & 1 & 0 & 1 & 0 & 1\\
\end{bmatrix}$. Therefore, $\bBd^{(1)}\bY_{\boldsymbol{\cdot}}^{(1)}$ represents the global-shared covariate effect, $\bBd^{(2)}\bY_{\boldsymbol{\cdot}}^{(2)},\bBd^{(3)}\bY_{\boldsymbol{\cdot}}^{(3)}$ are the row-shared covariate effects for view 1 and 2, $\bBd^{(4)}\bY_{\boldsymbol{\cdot}}^{(4)},\bBd^{(5)}\bY_{\boldsymbol{\cdot}}^{(5)}$ are the column-shared covariate effects for cohort 1 and 2, and $\bBd^{(6)}\bY_{\boldsymbol{\cdot}}^{(6)},\bBd^{(7)}\bY_{\boldsymbol{\cdot}}^{(7)},\bBd^{(8)}\bY_{\boldsymbol{\cdot}}^{(8)}, \bBd^{(9)}\bY_{\boldsymbol{\cdot}}^{(9)}$ are the individual covariate effects for $\bX_{11}, \bX_{21}, \bX_{12}, \bX_{22}$ respectively. The effect construction are similar for augmented effects $\bSddl, l=1,...,9$. $\bBdk$ and $\bSddl$ are generated to have true rank of 1 and 5 respectively to mimic the real-world scenario that a good amount of signals may not be captured by existing covariates. We first generate $\bBdk, \bSddl, \bY_j, \bE_{\boldsymbol{\cdot}\boldsymbol{\cdot}}$. Each matrix is of one unit standard deviation and therefore generate $ \bXdd = \sum^{9}_{k=1} a_k \bBdk \bY_{\boldsymbol{\cdot}}^{(k)} + \sum^{9}_{l=1}  b_l \bS_{\boldsymbol{\cdot}\boldsymbol{\cdot}}^{(l)} + \bE_{\boldsymbol{\cdot}\boldsymbol{\cdot}} $, where coefficients $a_k,b_l$ control the relative signal level in each case. By default, all coefficients are set as one.

In the estimation process, we carry out both optimization algorithms and their estimates are effectively identical when setting the same convergence criterion. For \cref{ALSUV_barrr} we set $r_{B,upper} = 7$ and $r_{S,upper} = 10$, respectively. Throughout, accuracy is reported as relative squared error (RSE): for a target matrix $\bA$ and its estimate $\widehat{\bA}$, $RSE(\bA) = ||\bA - \widehat{\bA}||_F^2 / ||\bA||_F^2$. For signal recovery the targets are the total signal, the summed covariate effects $\sum_k \bBdk\bYd^{(k)}$, the summed auxiliary structures $\sum_l \bSddl$, and each individual $\bBdk$ and $\bSddl$, the last two averaged ove modules. For imputation the sums run only over the set $\mathcal{M}$ of held-out entries, with the complete data $\bXdd$, as generated before missingness is introduced, as the target; an RSE of 1 then corresponds to imputing every held-out entry as zero.

\subsection{Signal recovery}
Table~\ref{paper3_tab: signal} reports relative squared errors (RSE) for the total signal, the summed covariate effects, the summed auxiliary structures, and the per-module averages for $\bBdk$ and $\bSddl$, as defined above. Both aggregate and module-level signals are recovered accurately when the coefficient matrices are low rank. Each component is estimated more precisely when its signal is stronger, and covariate effects of higher true rank are harder to recover. When the coefficients are not low rank relative to the covariate dimension (true rank 3 with 10 or 20 covariates), individual covariate modules become difficult to separate, although the aggregate covariate signal is still captured well.

\begin{table}[H]
\caption{Mean RSE of baRRR estimates over 100 replications, by the ratio $||\bB\bY||/||\bS||$ of covariate-effect to auxiliary-structure standard deviation, the true rank of each $\bBdk$, and the number of covariates in each $\bY_j$. Columns give RSE for the total signal, summed covariate effects, summed auxiliary structures, and the module averages for $\bBdk$ and $\bSddl$.}
\label{paper3_tab: signal}
\begin{tabular}{cccccccc}
\hline
$||\bB\bY||/||\bS||$ & rank  & dim  & RSE({$\bX$}) & RSE({$\bB\bY$}) & RSE({$\bS$}) & mean of & mean of \\ 
 &  of $\bB$ &  of $\bY$ &  &  &  &  RSE({$\bB$}) &  RSE({$\bS$}) \\ \hline
0.25 & 1         & 10       & 0.067  & 0.12    & 0.042  & 0.276        & 0.118        \\
     & 1         & 20       & 0.067  & 0.127   & 0.042  & 0.216        & 0.118        \\
     & 1         & 50       & 0.067  & 0.155   & 0.044  & 0.254        & 0.121        \\
     & 3         & 10       & 0.067  & 0.194   & 0.044  & 0.607        & 0.12         \\
     & 3         & 20       & 0.068  & 0.23    & 0.045  & 0.483        & 0.123        \\
     & 3         & 50       & 0.069  & 0.287   & 0.048  & 0.504        & 0.129        \\ \hline
1    & 1         & 10       & 0.133  & 0.028   & 0.14   & 0.133        & 0.251        \\
     & 1         & 20       & 0.134  & 0.03    & 0.141  & 0.076        & 0.252        \\
     & 1         & 50       & 0.134  & 0.036   & 0.146  & 0.072        & 0.259        \\
     & 3         & 10       & 0.136  & 0.057   & 0.153  & 0.473        & 0.265        \\
     & 3         & 20       & 0.137  & 0.07    & 0.159  & 0.269        & 0.275        \\
     & 3         & 50       & 0.139  & 0.087   & 0.165  & 0.207        & 0.29         \\ \hline
4    & 1         & 10       & 0.066  & 0.005   & 0.396  & 0.081        & 0.557        \\
     & 1         & 20       & 0.067  & 0.006   & 0.398  & 0.033        & 0.558        \\
     & 1         & 50       & 0.067  & 0.007   & 0.404  & 0.023        & 0.567        \\
     & 3         & 10       & 0.067  & 0.012   & 0.417  & 0.419        & 0.575        \\
     & 3         & 20       & 0.068  & 0.016   & 0.43   & 0.162        & 0.595        \\
     & 3         & 50       & 0.069  & 0.02    & 0.442  & 0.085        & 0.615  \\     \hline
\end{tabular}
\end{table}

\subsection{Missing data imputation}
\label{paper3_subsec:sim-missing}
We aim to evaluate the missing imputation performance of our proposed method in comparison with existing techniques. Operating under the assumption that data is missing completely at random, we  introduced four missing schemes as detailed below:

\begin{itemize}
\item \textbf{Partial Column Missing:} For a designated view, we introduce column-wise missingness independently to the samples. This results in varied columns missing across different views.
\item \textbf{Full Column Missing:} From the columns that originally have no missing data, we introduce complete column-wise missingness to certain samples across all views. This implies that specific columns will be absent across every view.
\item \textbf{Partial Row Missing:} For each cohort, we introduce row-wise missingness independently to the features. This leads to different rows being absent for various cohorts. Notably, we do not consider full row-wise missing akin to columns. This is because fully missing rows cannot be imputed; such an absence merely reduces the feature count, rendering prediction unfeasible.
\item \textbf{Entry Missing:} We introduce missingness randomly among entries that are not selected as missing in the first three schemes.
\end{itemize}

We allocate an equal proportion of missingness (m\% of the original data) at each step. Owing to the overlap of entries affected by both row and column missing schemes, the total missingness may be slightly reduced, but it should approximate 4m \%. These missing entries can be classified into six distinct categories:
\begin{enumerate}
    \item Entry only: selected by entry missing.
    \item Row only: selected by row missing only.
    \item Partial column only: selected by partial column missing only.
    \item Full column only: selected by full column missing only.
    \item Partial column and row: selected by both partial column missing and row missing.
    \item Full column and row: selected by both full column missing and row missing.
\end{enumerate}

We compare baRRR, which uses all 18 modules, with ten special cases that
keep only a subset of them (Table~\ref{paper3_tab:methods}), using the module
indexing above: 1 global, 2--3 row-shared, 4--5 column-shared, 6--9
individual. 

\begin{table}[H]
\centering
\small
\caption{Methods compared in the simulation, defined by which covariate-effect
($\bB\bY$) and auxiliary ($\bS$) modules each retains.}
\label{paper3_tab:methods}
\begin{tabular}{lccl}
\hline
Method & $\bB\bY$ modules & $\bS$ modules & Description \\ \hline
baRRR             & 1--9 & 1--9 & proposed method \\
BIDIFAC+          & --   & 1--9 & \citet{Lock2022BADIFACplus}; equals BIDIFAC \citep{park2020integrative} here \\
bRRR              & 1--9 & --   & covariate effects only \\
baRRR, row-shared & 2--3 & 2--3 & row-shared modules only \\
baRRR, col-shared & 4--5 & 4--5 & column-shared modules only \\
aRRR, one concat  & 1    & 1    & aRRR \citep{wang2023multiple} on $\bXdd$ \\
aRRR, four ind    & 6--9 & 6--9 & aRRR on each $\bX_{ij}$ \\
NNreg, one concat & 1    & --   & nuclear norm regression of $\bXdd$ on $\bYd$ \\
NNreg, four ind   & 6--9 & --   & nuclear norm regression of each $\bX_{ij}$ on $\bY_j$ \\
NNapp, one concat & --   & 1    & nuclear norm approximation of $\bXdd$ \\
NNapp, four ind   & --   & 6--9 & nuclear norm approximation of each $\bX_{ij}$ \\ \hline
\end{tabular}
\end{table}


The tabulated standard deviations (sd) in Table~\ref{paper3_tab: comparison}B represents variability of the imputation errors across different methodologies and missingness types. Similar values in each column (except sd=0 for RSE=1) indicate close standard deviations among different methods. Therefore, we only need to compare RSE for each method, which is shown in Table~\ref{paper3_tab: comparison}A. Less RSE indicates more accurate imputation. Our proposed method, baRRR, consistently outperforms all other methods across every type of missing data. Unsupervised methods, like BIDIFAC+ and nuclear norm approximation, are generally less accurate than supervised methods with identical module structures when columns are missing. The discrepancy in performance of all other methods sheds light on certain dependencies: full column recovery (i.e., signals for specific samples) is primarily based on covariate information, while the recovery of partially missing columns largely depends on row-shared structures. Row recovery, which refers to the signals of particular features, is influenced by both column-shared and global covariate effects as well as auxiliary structures. Notably, the performance gap between baRRR and BIDIFAC+, especially in the context of column missingness, underscores the pivotal role covariate effects play in data interpretation, particularly when columns are missing. Because a sample with all outcomes missing contributes no information to the auxiliary structures, its imputed values are predictions from its covariates alone; full column missingness therefore assesses prediction of outcomes for new samples.

\begin{table}[H]
\setlength{\tabcolsep}{2pt} 
\caption{Mean (\textbf{A}) and standard deviation (\textbf{B}) of imputation RSE over 100 replications for baRRR and 10 comparison methods (rows) under six types of missingness (columns); the last column pools all missing entries. About $4\%$ of $\bXdd$ is missing and each method runs for 100 epochs. Bold entries mark the smallest value in each column; values below $10^{-6}$ are shown as 0.}
\label{paper3_tab: comparison}
\begin{tabular}{lccccccc}
\hline
\multicolumn{1}{c}{\textbf{A}} & \multicolumn{1}{c}{entry} & \multicolumn{1}{c}{row}  & \multicolumn{1}{c}{partial col} & \multicolumn{1}{c}{full col} & \multicolumn{1}{c}{partial col} & \multicolumn{1}{c}{full col} & \multicolumn{1}{c}{all missing} \\
\multicolumn{1}{c}{} & \multicolumn{1}{c}{only}  & \multicolumn{1}{c}{only} & \multicolumn{1}{c}{only}        & \multicolumn{1}{c}{only}     & \multicolumn{1}{c}{and row}     & \multicolumn{1}{c}{and row}  & \multicolumn{1}{c}{together}    \\\hline
baRRR                & $\boldsymbol{0.108} $                    & $\boldsymbol{0.631} $                   & $\boldsymbol{0.374} $                          & $\boldsymbol{0.549}$                        & $\boldsymbol{0.732}$                           & $\boldsymbol{0.802} $                       & $\boldsymbol{0.420}$                            \\
BIDIFAC(+)           & 0.119                     & 0.654                    & 0.629                           & 1                            & 0.864                           & 1                            & 0.600                             \\
bRRR                 & 0.465                     & 0.776                    & 0.579                           & 0.613                        & 0.785                           & 0.813                        & 0.609                           \\
baRRR, row-shared    & 0.294                     & 1                        & 0.444                           & 0.574                        & 1                               & 1                            & 0.593                           \\
baRRR, col-shared    & 0.349                     & 0.728                    & 0.688                           & 0.689                        & 0.818                           & 0.842                        & 0.613                           \\
aRRR, one concat     & 0.398                     & 0.742                    & 0.544                           & 0.694                        & 0.811                           & 0.84                         & 0.596                           \\
aRRR, four ind       & 0.241                     & 1                        & 0.563                           & 0.574                        & 1                               & 1                            & 0.608                           \\
NNreg, one concat    & 0.637                     & 0.861                    & 0.663                           & 0.675                        & 0.854                           & 0.894                        & 0.710                            \\
NNreg, four ind      & 0.498                     & 1                        & 0.581                           & 0.592                        & 1                               & 1                            & 0.678                           \\
NNapp, one concat    & 0.488                     & 0.789                    & 0.772                           & 1                            & 0.948                           & 1                            & 0.760                            \\
NNapp, four ind      & 0.349                     & 1                        & 1                               & 1                            & 1                               & 1                            & 0.842                \\   \hline       
\end{tabular}

\vspace{20 pt}

\begin{tabular}{lccccccc}
\hline
\multicolumn{1}{c}{\textbf{B}} & \multicolumn{1}{c}{entry} & \multicolumn{1}{c}{row}  & \multicolumn{1}{c}{partial col} & \multicolumn{1}{c}{full col} & \multicolumn{1}{c}{partial col} & \multicolumn{1}{c}{full col} & \multicolumn{1}{c}{all missing} \\
\multicolumn{1}{c}{} & \multicolumn{1}{c}{only}  & \multicolumn{1}{c}{only} & \multicolumn{1}{c}{only}        & \multicolumn{1}{c}{only}     & \multicolumn{1}{c}{and row}     & \multicolumn{1}{c}{and row}  & \multicolumn{1}{c}{together}    \\ \hline
baRRR                & 0.01                      & 0.08                     & 0.06                            & 0.10                          & 0.22                            & 0.23                         & 0.03                            \\
BIDIFAC(+)           & 0.01                      & 0.06                     & 0.07                            & 0                            & 0.15                            & 0                            & 0.04                            \\
bRRR                 & 0.03                      & 0.08                     & 0.10                             & 0.12                         & 0.22                            & 0.24                         & 0.04                            \\
baRRR, row-shared    & 0.02                      & 0                        & 0.07                            & 0.09                         & 0                               & 0                            & 0.04                            \\
baRRR, col-shared    & 0.02                      & 0.09                     & 0.08                            & 0.09                         & 0.23                            & 0.28                         & 0.04                            \\
aRRR, one concat     & 0.02                      & 0.09                     & 0.09                            & 0.09                         & 0.33                            & 0.29                         & 0.04                            \\
aRRR, four ind       & 0.01                      & 0                        & 0.08                            & 0.09                         & 0                               & 0                            & 0.04                            \\
NNreg, one concat    & 0.03                      & 0.11                     & 0.10                             & 0.11                         & 0.27                            & 0.34                         & 0.04                            \\
NNreg, four ind      & 0.03                      & 0                        & 0.09                            & 0.11                         & 0                               & 0                            & 0.04                            \\
NNapp, one concat    & 0.03                      & 0.08                     & 0.10                             & 0                            & 0.31                            & 0                            & 0.04                            \\
NNapp, four ind      & 0.02                      & 0                        & 0                               & 0                            & 0                               & 0                            & 0.02           \\       \hline         
\end{tabular}

\end{table}
Given the flexibility of baRRR in recovering different module structures, we
further vary the signal levels $a_k$ and $b_l$ across six configurations. In
the standard configuration all equal 1. In each of the others, the
highlighted modules have standard deviation 2 and all remaining modules 0.5:
all covariate effects (high BY; $a_1,\dots,a_9$), the global auxiliary module
(high S global; $b_1$), the row-shared auxiliary modules (high S row;
$b_2,b_3$), the column-shared ones (high S column; $b_4,b_5$), or the
individual ones (high S individual; $b_6,\dots,b_9$).


We present the imputation relative squared error (RSE) for baRRR under various signal levels and missing proportions (\cref{paper3_tab:miss_rank1} when true rank of $\bBdk$ is 1 and \cref{paper3_tab:miss_rank3} when true rank of $\bBdk$ is 3).  The accuracy of missing imputation diminishes with an increasing overall missing rate. Interestingly, the performance remains relatively consistent for missing rates up to 20\%. The accuracy tends to decrease as the true rank of the covariate effect grows. Our method adeptly captures the signal for each structure when these structures exhibit even signals, especially more precise when there is a higher coefficient effect: in the high BY configuration at a 4\% missing rate, RSE for full column entries falls from 0.549 to 0.082 relative to the standard configuration. 
The signal of specific structures is captured effectively as long as it demonstrates stronger signals compared to other components. This is evident from the low RSE observed in partial column-only and row-only recovery when there exists a high global $\bS$ signal, partial column-only for a high column $\bS$ signal, and row-only for a high row $\bS$ signal. Such findings underscore the adaptability of our method, showcasing its ability to efficiently handle missing imputation across a spectrum of signal scenarios.

\begin{table}[H]
\setlength{\tabcolsep}{2pt} 
\caption{Mean imputation RSE of baRRR over 100 replications, by signal configuration and missing rate, when each $\bBdk$ has true rank 1.}
\label{paper3_tab:miss_rank1}
\centering
\begin{tabular}{lcccccccc}
\hline
Signal        & \multicolumn{1}{c}{Missing}     & \multicolumn{1}{c}{entry} &  \multicolumn{1}{c}{full col} & \multicolumn{1}{c}{full col} & \multicolumn{1}{c}{partial col} & \multicolumn{1}{c}{partial col} & 
\multicolumn{1}{c}{row}  &\multicolumn{1}{c}{all missing} \\
levels        & \multicolumn{1}{c}{rate} & \multicolumn{1}{c}{only}  &  \multicolumn{1}{c}{only}        & \multicolumn{1}{c}{and row}     & \multicolumn{1}{c}{only}     & \multicolumn{1}{c}{and row}  & 
\multicolumn{1}{c}{only} &\multicolumn{1}{c}{together}    \\ \hline
standard      & 0.04                            & 0.108                     & 0.549                    & 0.802                           & 0.374                        & 0.732                           & 0.631                        & 0.420                           \\
              & 0.08                            & 0.115                     & 0.530                    & 0.777                           & 0.378                        & 0.716                           & 0.627                        & 0.411                           \\
              & 0.2                             & 0.138                     & 0.546                    & 0.771                           & 0.397                        & 0.726                           & 0.649                        & 0.436                           \\
              & 0.4                             & 0.196                     & 0.572                    & 0.797                           & 0.451                        & 0.764                           & 0.689                        & 0.480                           \\
              & 0.8                             & 0.464                     & 0.706                    & 0.868                           & 0.661                        & 0.860                           & 0.829                        & 0.666                           \\ \hline
high BY       & 0.04                            & 0.035                     & 0.082                    & 0.583                           & 0.064                        & 0.548                           & 0.535                        & 0.189                           \\
              & 0.08                            & 0.037                     & 0.076                    & 0.543                           & 0.065                        & 0.531                           & 0.532                        & 0.177                           \\
              & 0.2                             & 0.042                     & 0.085                    & 0.530                           & 0.071                        & 0.531                           & 0.535                        & 0.190                           \\
              & 0.4                             & 0.054                     & 0.120                    & 0.556                           & 0.101                        & 0.564                           & 0.549                        & 0.210                           \\
              & 0.8                             & 0.109                     & 0.264                    & 0.619                           & 0.229                        & 0.617                           & 0.603                        & 0.307                           \\ \hline
high S  & 0.04                            & 0.108                     & 0.846                    & 0.935                           & 0.222                        & 0.437                           & 0.330                        & 0.371                           \\
   global           & 0.08                            & 0.115                     & 0.844                    & 0.925                           & 0.230                        & 0.408                           & 0.328                        & 0.379                           \\
              & 0.2                             & 0.134                     & 0.859                    & 0.931                           & 0.250                        & 0.437                           & 0.357                        & 0.403                           \\
              & 0.4                             & 0.178                     & 0.874                    & 0.946                           & 0.308                        & 0.506                           & 0.409                        & 0.452                           \\
              & 0.8                             & 0.387                     & 0.939                    & 0.977                           & 0.590                        & 0.763                           & 0.650                        & 0.653                           \\ \hline
high S     & 0.04                            & 0.115                     & 0.856                    & 0.936                           & 0.249                        & 0.917                           & 0.905                        & 0.550                            \\
 column             & 0.08                            & 0.122                     & 0.849                    & 0.932                           & 0.253                        & 0.920                           & 0.905                        & 0.530                            \\
              & 0.2                             & 0.142                     & 0.857                    & 0.929                           & 0.280                        & 0.927                           & 0.915                        & 0.552                           \\
              & 0.4                             & 0.190                     & 0.875                    & 0.945                           & 0.348                        & 0.943                           & 0.931                        & 0.585                           \\
              & 0.8                             & 0.414                     & 0.940                    & 0.977                           & 0.661                        & 0.976                           & 0.972                        & 0.742                           \\ \hline
high S     & 0.04                            & 0.117                     & 0.853                    & 0.928                           & 0.815                        & 0.920                           & 0.360                        & 0.528                           \\
  row            & 0.08                            & 0.123                     & 0.848                    & 0.922                           & 0.824                        & 0.923                           & 0.359                        & 0.541                           \\
              & 0.2                             & 0.143                     & 0.860                    & 0.932                           & 0.838                        & 0.932                           & 0.391                        & 0.561                           \\
              & 0.4                             & 0.190                     & 0.874                    & 0.945                           & 0.860                        & 0.943                           & 0.454                        & 0.605                           \\
              & 0.8                             & 0.407                     & 0.939                    & 0.976                           & 0.935                        & 0.975                           & 0.700                        & 0.755                           \\ \hline
high S     & 0.04                            & 0.126                     & 0.854                    & 0.950                           & 0.821                        & 0.916                           & 0.908                        & 0.690                            \\
 individual             & 0.08                            & 0.132                     & 0.850                    & 0.931                           & 0.824                        & 0.923                           & 0.904                        & 0.675                           \\
              & 0.2                             & 0.154                     & 0.863                    & 0.934                           & 0.836                        & 0.927                           & 0.918                        & 0.692                           \\
              & 0.4                             & 0.203                     & 0.876                    & 0.945                           & 0.862                        & 0.944                           & 0.932                        & 0.710                           \\
              & 0.8                             & 0.430                     & 0.939                    & 0.977                           & 0.935                        & 0.976                           & 0.972                        & 0.806                          \\\hline
\end{tabular}
\end{table}

\section{Application}
\label{paper3_sec:application}
\subsection{Cancer type, mutation and other data details}
\label{paper3_subsec:names}


We use the TCGA pan-cancer data of \cite{hoadley2018cell}. The outcomes comprise four omics sources: batch-corrected mRNA expression for the $p_1 = 1000$ genes with the highest standard deviation among 20531, batch-corrected miRNA-Seq data for $p_2 = 743$ miRNAs, DNA methylation for the $p_3 = 1000$ CpG sites with the highest standard deviation among 22601, and batch-corrected reverse-phase protein array data for $p_4 = 198$ proteins. The mRNA and miRNA data are log-transformed. Each feature is centered to mean 0 across all samples, and each source $\bX_{i\boldsymbol{\cdot}}$ is scaled to unit standard deviation to put the sources on a common scale. As mutation covariates we use binary indicators for the 50 most frequently mutated genes, coded 1 if a somatic mutation was called in that gene for a given tumor and 0 otherwise. Matching samples across all data types leaves $n = 6165$ tumors from $J = 27$ cancer types, summarized in \cref{paper3_append_caner_names}. Overall, 2.9\% of $\bXdd$ is missing, in both partial rows and partial columns.

\subsection{Decomposition results}

We use $H=2$ covariate views. The first holds the 27 cancer-type indicators and enters as a single module shared across all sources and cohorts, capturing baseline differences between cancer types. The second holds the 50 mutation indicators and enters through $K_2=50$ modules. The auxiliary structures use the same 50 modules, so that within each module baRRR separates mutation-driven from mutation-unrelated variation. The modules are the 50 combinations of sources and cohorts with the largest variation found by BIDIFAC+ \citep{Lock2022BADIFACplus}, shown in
Figure~\ref{paper3_fig:heatmap_module}. The fitted model is \begin{align}\label{newmodel}
    \bXdd &= {\mathbf{B}_{\boldsymbol{\cdot}1}^{}} \bY_{1\boldsymbol{\cdot}}^{} +  \sum^{50}_{k=1} {\mathbf{B}_{\boldsymbol{\cdot}2}^{(k)}} \bY_{2\boldsymbol{\cdot}}^{(k)} + \sum^{50}_{k=1}  {\mathbf{S}_{\boldsymbol{\cdot\cdot}}^{(k)}} + \bEdd,
    \end{align}
where $\bY_{1\boldsymbol{\cdot}}$ has 27 indicator features, with 1/0 indicating whether a sample belongs to the current cancer cohort. We orthogonalize $\bY_{1\boldsymbol{\cdot}},\bY_{2\boldsymbol{\cdot}}^{(1)},...,\bY_{2\boldsymbol{\cdot}}^{(50)}$ before optimization.

\begin{figure}[H] 
\caption{The 50 modules selected by BIDIFAC+ and used for both the mutation-driven and auxiliary structures. Rows are modules and columns the 27 cancer types; a colored cell means the cancer type belongs to the module, and the colors show which omics sources the module spans.}
\includegraphics[width=\textwidth,height=16cm,keepaspectratio]{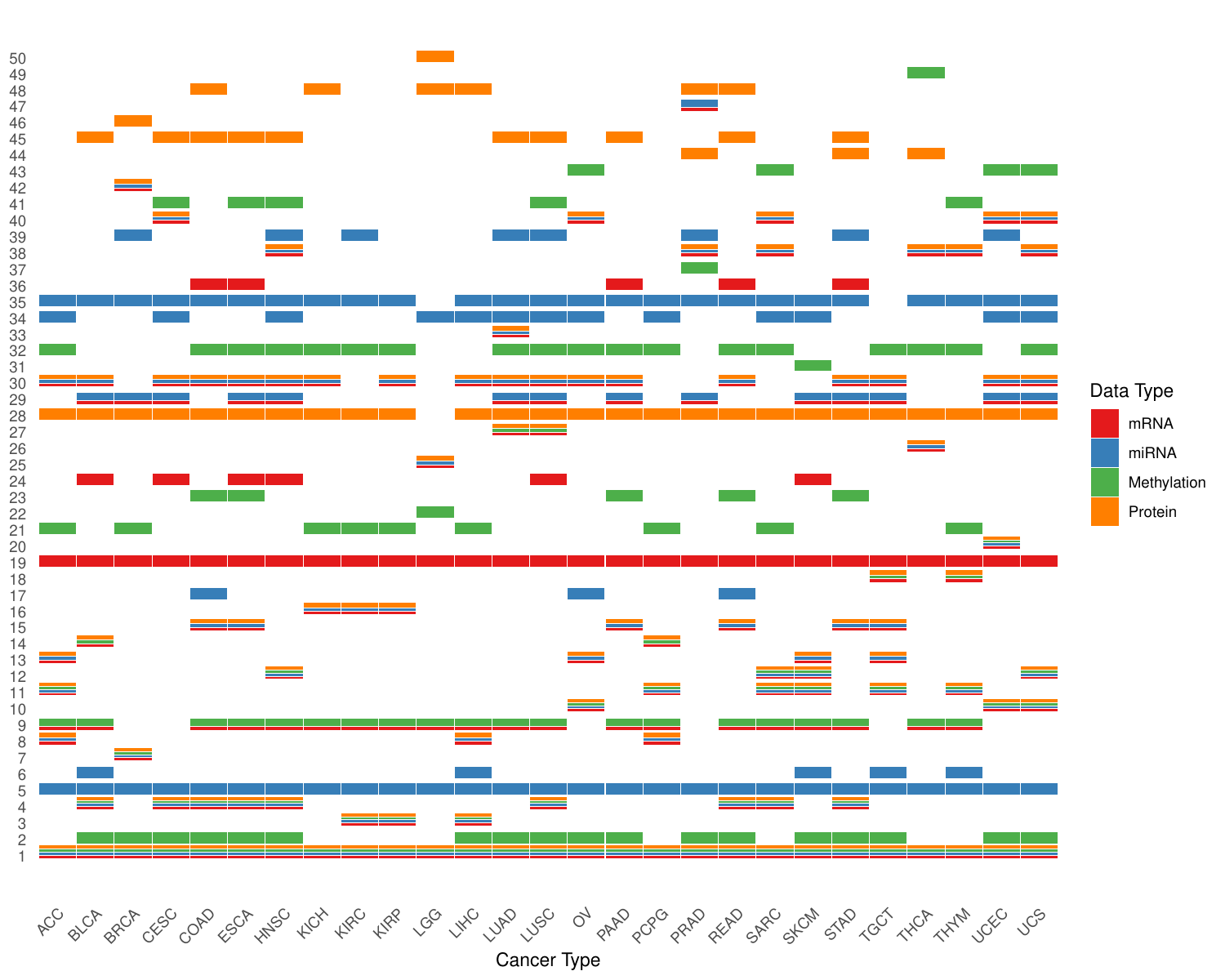}
\label{paper3_fig:heatmap_module}   
\end{figure}

Since our proposed baRRR method accounts for missing data, we apply  \cref{ALSUV_barrr} with rank upper bound 20 (i.e.$r_{B,upper} = r_{S,upper} = 20$) directly to estimate ${\mathbf{B}_{\boldsymbol{\cdot}1}^{}},{\mathbf{B}_{\boldsymbol{\cdot}2}^{(1)}},...,{\mathbf{B}_{\boldsymbol{\cdot}2}^{(50)}}, {\mathbf{S}_{\boldsymbol{\cdot\cdot}}^{(1)}},...,{\mathbf{S}_{\boldsymbol{\cdot\cdot}}^{(50)}} $ while simultaneously imputing the missing values. As illustrated in \cref{paper3_fig:scatter_signal}, the scatter plots comparing the original data matrix $\bXdd$ (with missing values filled with 0) and the total recovered signals $\widehat{\bXdd}$ demonstrate striking similarities. This resemblance suggests that our baRRR approach has adeptly captured the principal signals in the TCGA data.  We assess the extent of the signal captured using the relative unexplained variance (RUV), defined as a function of $\bA$, i.e. $RUV(\bA) = \frac{||\bXdd - \widehat{\bXdd} + \hat{\bA}||_F^2}{||\bXdd||_F^2}$, where $\hat{\bA}$ symbolizes any estimated component of the model. The purpose of RUV is to determine the remaining variation when $\hat{\bA}$ is left out of the model. With an RUV of 0.32 for the baRRR model (assuming $\hat{\bA} = \boldsymbol{0}$), it's evident that the model successfully encapsulates the primary signals inherent in the original data. When key structures such as ${\mathbf{B}_{\boldsymbol{\cdot}1}^{}} \bY_{1\boldsymbol{\cdot}}^{}, \sum^{50}_{k=1} {\mathbf{B}_{\boldsymbol{\cdot}2}^{(k)}} \bY_{2\boldsymbol{\cdot}}^{(k)}, \sum^{50}_{k=1}  {\mathbf{S}_{\boldsymbol{\cdot\cdot}}^{(k)}}$ are omitted, the RUV increases to 0.49, 0.35 and 0.72 respectively, underscoring that augmented structures explain the most variation followed by cohort effects and then mutation effects. 

Table~\ref{paper3_tab:top15mods} lists the 15 modules explaining the most
variation in the baRRR fit. The three largest are a global module spanning
all cancer types and sources, a methylation module shared by 18 cancer
types, and a module shared by the kidney cancers KIRC and KIRP together with
liver cancer (LIHC). Mutation-driven variation is small relative to
auxiliary structure in most modules, but its share varies sharply with the
cancer types a module spans. The broad modules shared by all or nearly all
cancer types (modules 1, 2, 5 and 9) carry less than 0.5\% of their
variation in $\bB\bY$. The modules with substantial mutation-driven shares
each span a small group of three to six cancer types: 47\% for module 15
(gastrointestinal cancers and TGCT), 37\% for module 11, 21\% for module 10
(OV, UCEC, UCS) and 9\% for module 13, although not every such module shows
this (modules 3 and 7 carry almost none). In these data, somatic mutations
thus act mainly on structure specific to groups of cancer types rather than
on pan-cancer structure, extending to four omics sources the finding of
\cite{wang2023multiple} for gene expression. The dominance of auxiliary
variation also shows why it must be modeled: without it, that variation
would be left in the error term or could be absorbed into the mutation
effects.

Our estimates pave the way for biologically significant discoveries. For example, principal components plots of module 10 (OV, UCEC, UCS) are shown in the lower panels of Figure~\ref{paper3_fig:scatter_signal}. These illustrate the mutation-associated variation for a shared effect between ovarian and uterine cancers, with module 10's first two principle components categorized based on the presence of mutations in either TP53 or PTEN genes. It's well-documented that TP53 mutations are common and correlate with a worse prognosis in ovarian and uterine cancers. Additionally, PTEN mutations, particularly prevalent in UCEC, are connected to an increased rate of mutations in cancers that are medium to poorly differentiated, notably in younger women \citep{tao2020pten}.

\begin{figure}[!ht]
    \centering   \includegraphics[height=5.5cm,width=7cm]{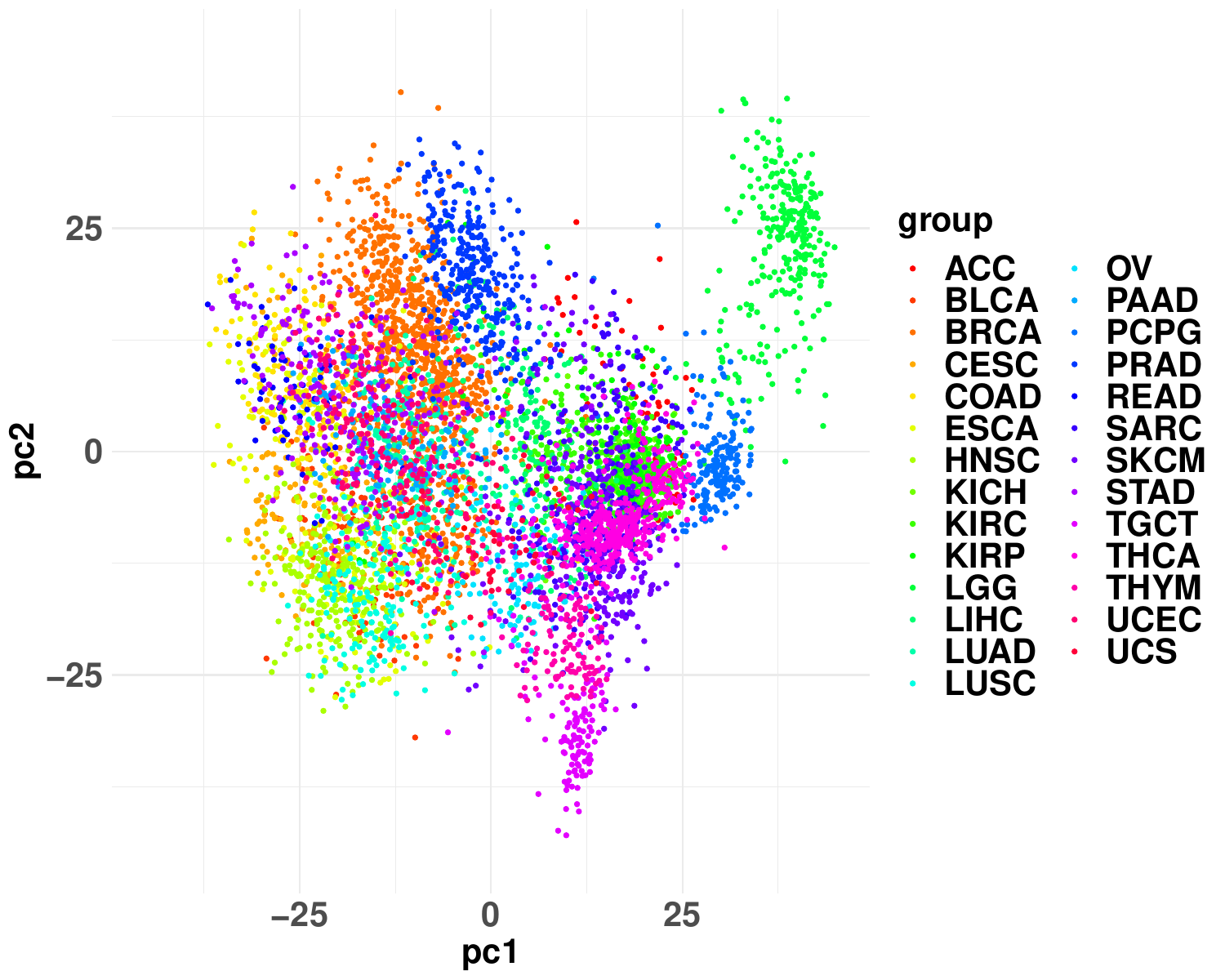}
    \includegraphics[height=5.5cm,width=7cm]{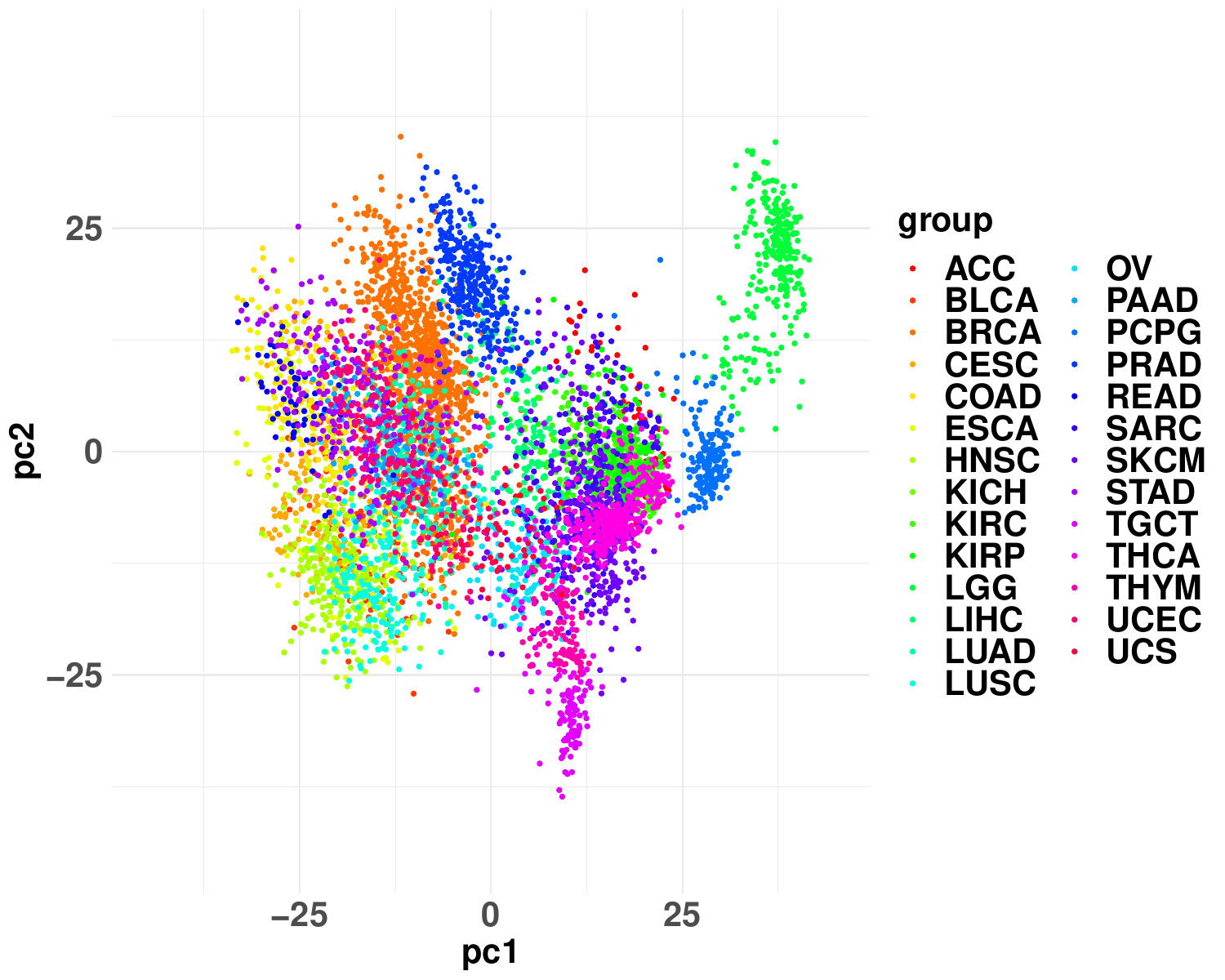}\\
    \includegraphics[height=5.5cm,width=7cm]{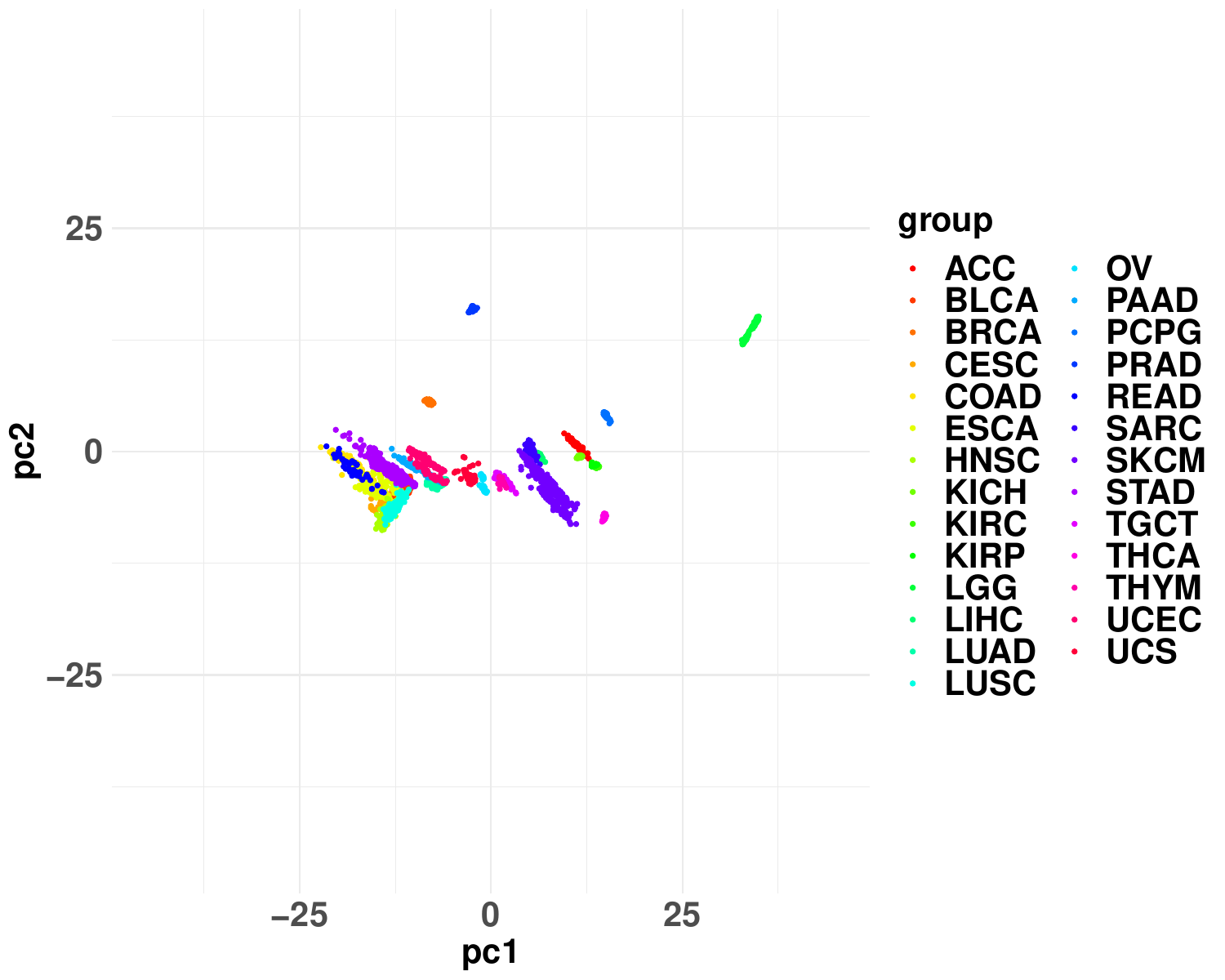}
    \includegraphics[height=5.5cm,width=7cm]{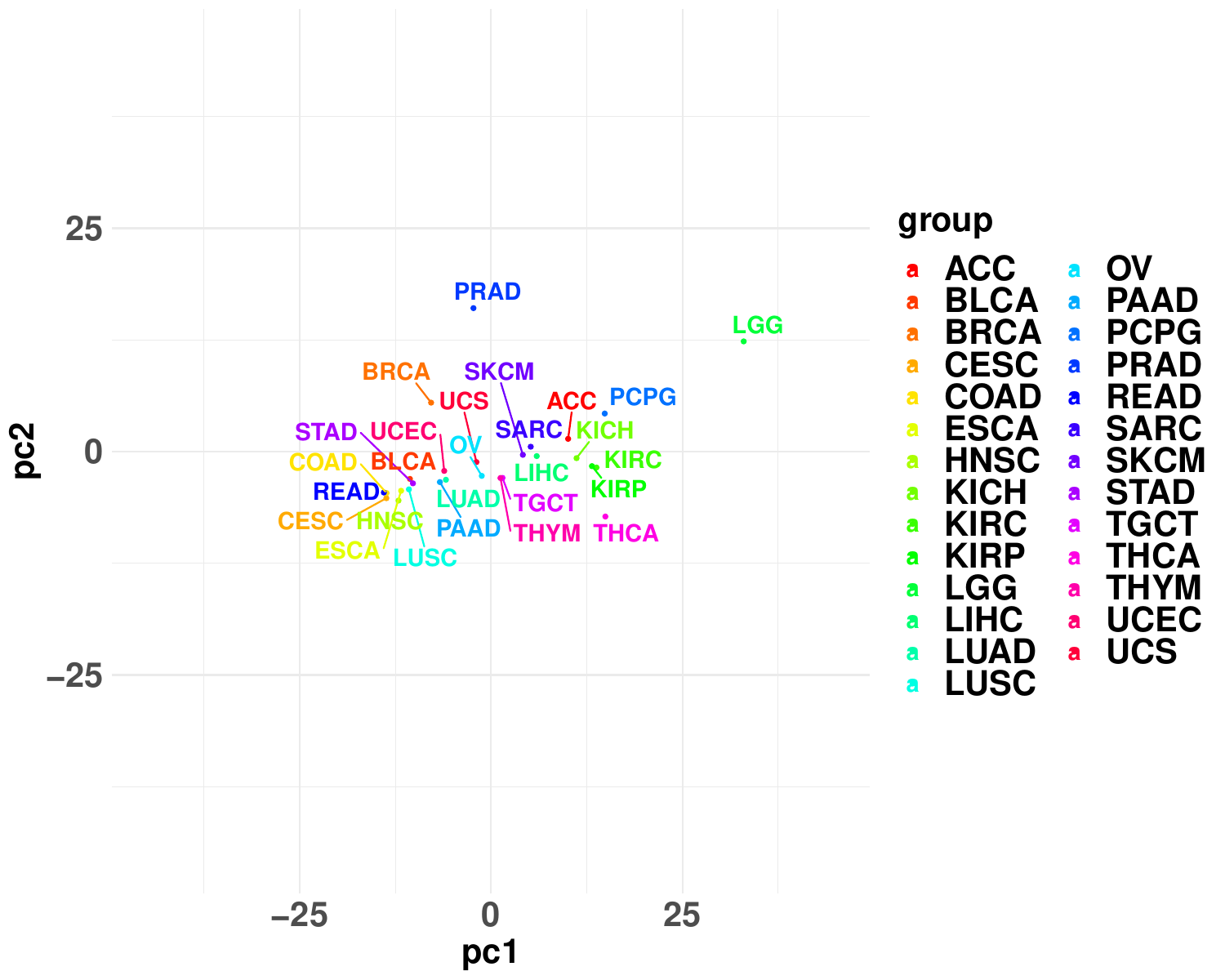}\\
    \includegraphics[height=5.5cm,width=7cm]{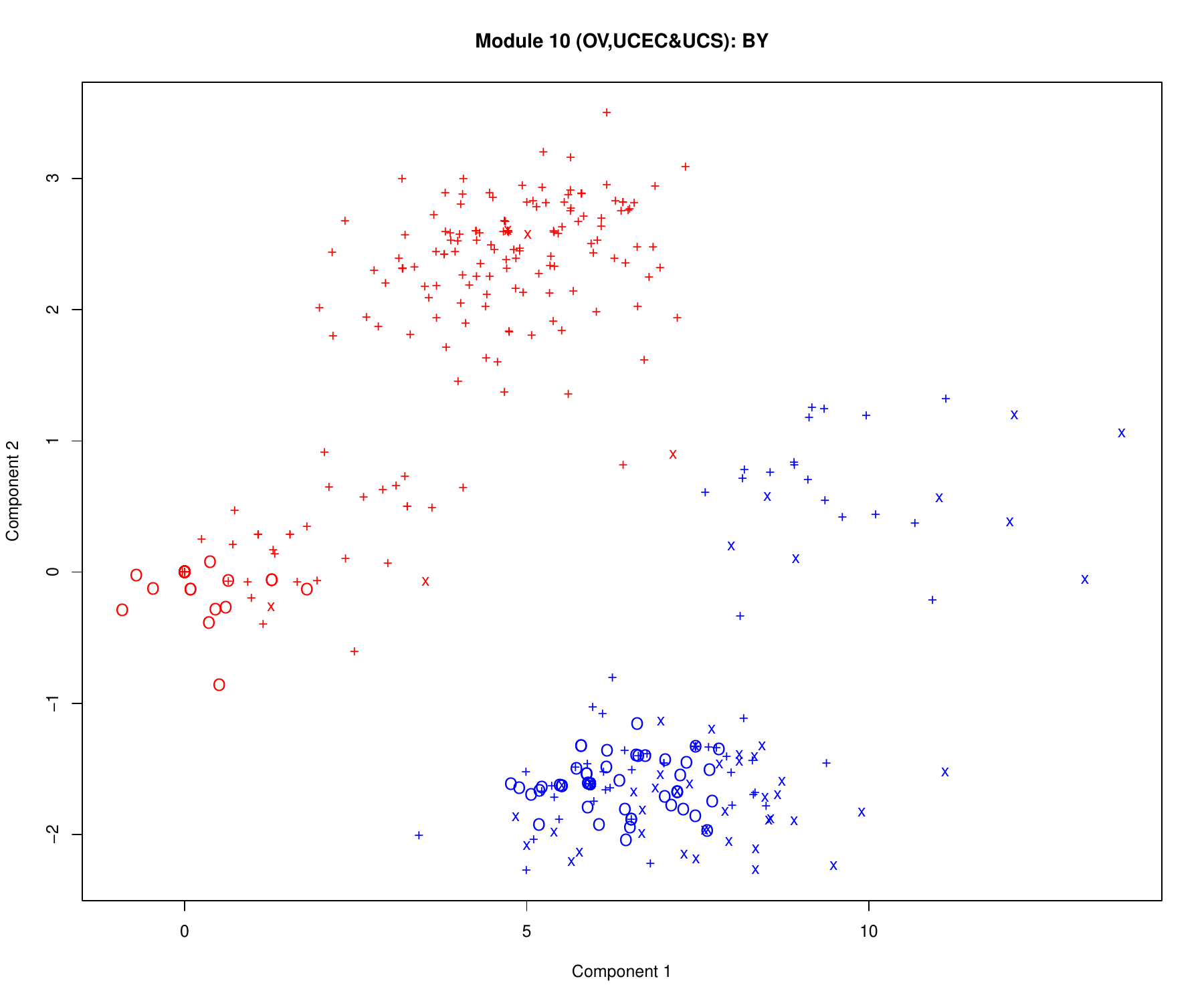}
    \includegraphics[height=5.5cm,width=7cm]{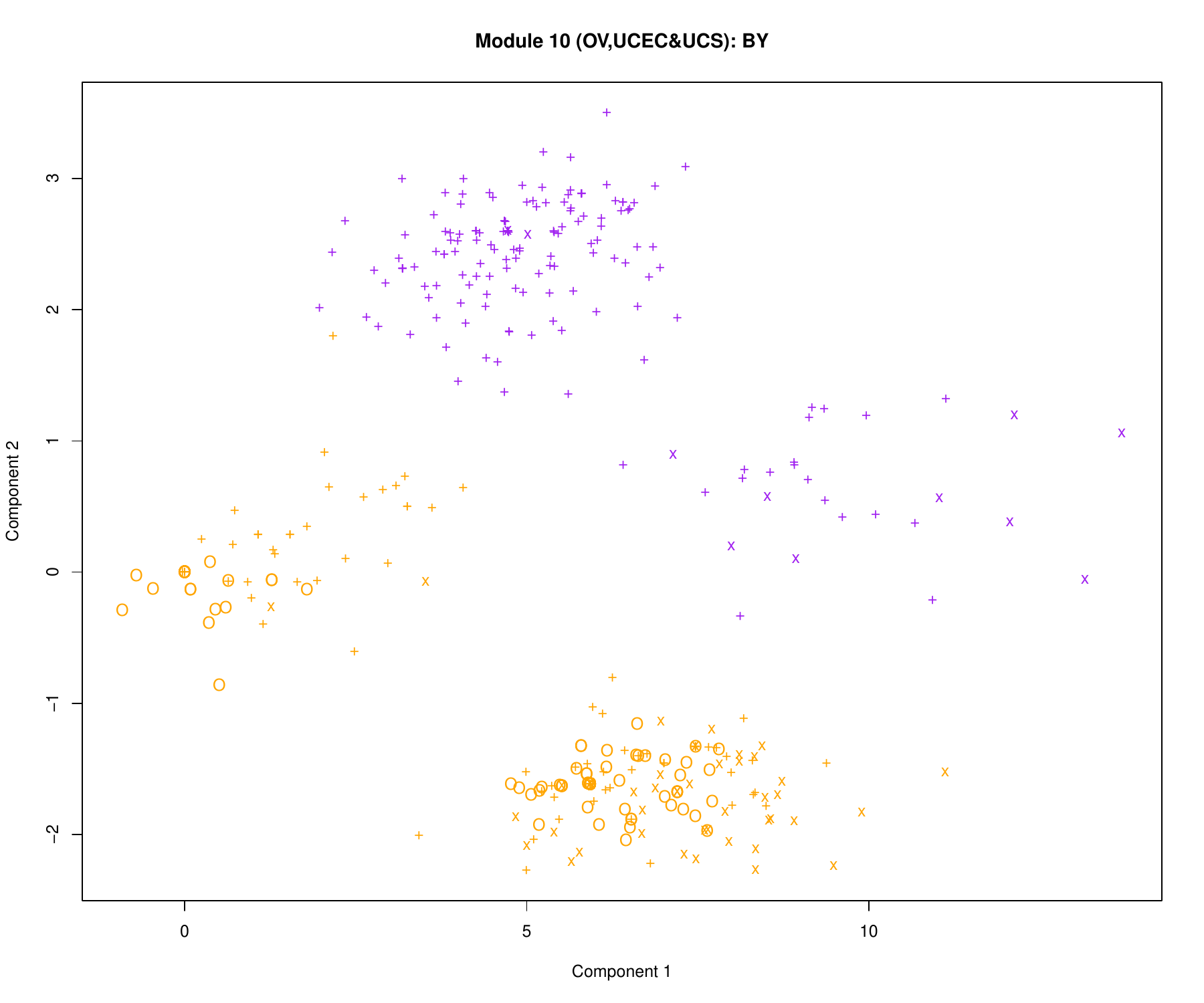}
    \caption{$\bold{Upper-left}$: first two principal component scores (PCs) of original data $\bXdd$. $\bold{Upper-right}$:  first two PCs of total signals $\widehat{\bXdd}$ recovered by baRRR. Each color represents a cancer cohort, while each dot represents a sample. $\bold{Middle-left}$: first two PCs of mutation effects plus cohort effects $\bB_{\cdot1}\bY_{1\cdot} + \sum_{k=1}^{50} \bB_{\cdot2}^{(k)}\bY_{2\cdot}^{(k)}$. $\bold{Middle-right}$: first two PCs of cohort effects only $\bB_{\cdot1}\bY_{1\cdot}$ recovered by baRRR. Each color represents a cancer cohort, while each dot represents one feature.  $\bold{Lower-left}$: first two PCs of covariate-related variation ($\bB \bY)$ from Module 10 (OV,UCEC,UCS), with colors showing TP53 mutation (red means without and blue means with TP53 mutation).  $\bold{Lower-right}$: the first two PCs of covariate-related variation ($\bB \bY)$ from Module 10 (OV,UCEC,UCS), with color showing PTEN mutation (orange: without, purple: with) and symbols cancer type (``O": OV, ``+": UCEC, ``x": UCS. }
\label{paper3_fig:scatter_signal}
\end{figure}

\begin{table}[!h]
\caption{Cancer types, molecular modalities and sources for the top 15 modules, ordered by variation explained by baRRR.}
\label{paper3_tab:top15mods}
\setlength{\tabcolsep}{2pt} 
\begin{tabular}{cccccc}
\hline
Module & Variance  & Variance  & Cancer  & Sample  \\
 & of BY & of S &  types &  size \\
\hline
1      & 368           & 221163        & All cancers                                                & 6165   \\
2      & 77            & 176924        & All cancers but *ACC,KICH,KIRC,KIRP, & 4220   \\
       &               &               & LGG,PCPG,SARC,THCA,THYM* &    \\
3      & 64            & 135178        & KIRC,KIRP,LIHC                                             & 771    \\
4      & 3069          & 126379        & BLCA,CESC,COAD,READ,ESCA,HNSC, & 1694   \\
       &               &               & LUSC,SARC,STAD &     \\
5      & 0             & 127126        & All cancers                                                & 6165   \\
6      & 1557          & 119073        & BLCA,LIHC,SKCM,TGCT,THYM                                   & 918    \\
7      & 0             & 106231        & BRCA                                                       & 976    \\
8      & 921           & 88035         & ACC,LIHC,PCPG                                              & 451    \\
9      & 347           & 87474         & All cancers but *BRCA,CESC,OV,PRAD, & 4118   \\
       &               &               & TGCT,UCEC,UCS* &     \\
10     & 15628         & 58334         & OV,UCEC,UCS                                                & 414    \\
11     & 26083         & 44557         & ACC,PCPG,SARC,SKCM,TGCT,THYM                               & 1095   \\
12     & 695           & 62386         & HNSC,SARC,SKCM,UCS                                         & 923    \\
13     & 5592          & 56809         & ACC,OV,SKCM,TGCT                                           & 667    \\
14     & 2082          & 56126         & BLCA,PCPG                                                  & 308    \\
15     & 25915         & 29106         & COAD,READ,ESCA,PAAD,STAD,TGCT                              & 953    \\
\hline
\end{tabular}

\vspace{2cm}

\begin{tabular}{cccc}
\hline
Module & Molecular modalities & Feature size \\
\hline
1      & mRNA,miRNA,Methylation,Protein & 2941  \\
2      & Methylation                    & 1000  \\
3      & mRNA,miRNA,Methylation,Protein & 2941  \\
4      & mRNA,miRNA,Methylation,Protein & 2941  \\
5      & miRNA                          & 743   \\
6      & miRNA                          & 743   \\
7      & mRNA,miRNA,Methylation,Protein & 2941  \\
8      & mRNA,miRNA,Protein             & 1941  \\
9      & mRNA,Methylation               & 2000  \\
10     & mRNA,miRNA,Methylation,Protein & 2941  \\
11     & mRNA,miRNA,Methylation,Protein & 2941  \\
12     & mRNA,miRNA,Methylation,Protein & 2941  \\
13     & mRNA,miRNA,Protein             & 1941  \\
14     & mRNA,Methylation,Protein       & 2198  \\
15     & mRNA,miRNA,Protein             & 1941  \\
\hline
\end{tabular}
\end{table}

\subsection{Missing Data Imputation}


We assess imputation on the real data by adding missingness to the 2.9\% already present, using the four schemes (full column, partial column, row, and random entry missing) of Section~\ref{paper3_subsec:sim-missing}, each targeting an additional 3\% of entries. Because new entries can overlap with each other and with the existing missing data, total missingness is about 14.5\%. RSE is computed over the artificially removed entries only, and results are averaged over 10 independent missingness patterns.

We compare baRRR with ten alternatives that drop or merge parts of its
structure (Table~\ref{paper3_tab:methods_real}). Stacked methods concatenate
the four omics sources into a single view, so that modules are defined over
cancer types only; ``27 ind'' denotes one separate module per cancer type.

\begin{table}[H]
\centering
\small
\caption{Methods compared on the TCGA data. Entries give the number of
modules for the mutation effect, the cohort effect (the cancer-type
indicator module $\bB_{\boldsymbol{\cdot}1}\bY_{1\boldsymbol{\cdot}}$ of
model~\eqref{newmodel}) and the auxiliary structure; ``--'' means the term
is omitted.}
\label{paper3_tab:methods_real}
\begin{tabular}{lcccc}
\hline
Method & Mutation & Cohort & Auxiliary & Sources \\
       & effect   & effect & structure &         \\ \hline
baRRR                    & 50          & 1  & 50          & separate \\
BIDIFAC+                 & --          & -- & 50          & separate \\
bRRR                     & 50          & 1  & --          & separate \\
baRRR w/o cohort         & 50          & -- & 50          & separate \\
maRRR                    & 50          & -- & 50          & stacked \\
aRRR                     & 1 (global)  & 1  & 1 (global)  & stacked \\
aRRR, 27 ind             & 27 (ind)    & 1  & 27 (ind)    & stacked \\
NNapp, 27 ind            & --          & 1  & 27 (ind)    & stacked \\
NNapp, 27 ind w/o cohort & --          & -- & 27 (ind)    & stacked \\
NNreg, 27 ind            & 27 (ind)    & 1  & --          & stacked \\
NNreg, 27 ind w/o cohort & 27 (ind)    & -- & --          & stacked \\ \hline
\end{tabular}
\end{table}

As shown in \cref{paper3_tab:miss_real_data} baRRR achieved the lowest or nearly lowest RSE in all six types of missing data structures and the lowest RSE overall. Full/partial column and row cases have less frequency since they are intersections of column and row missingness, leading to more variation. Excluding the cohort effect from the baRRR model significantly increased the RSE in the full column missing scenario, indicating that different cohorts contain unique information. The performance of maRRR (i.e., further ignoring information in different views) suggests that essential information is captured through modules from different cohorts. BIDIFAC+ indicates that missing covariates will lose view-specific information, while mRRR shows that row and entry-wise missing data recovery primarily relies on augmented structures. aRRR demonstrates that many signals can be merged for imputation. The 27 individual nuclear norm regression or approximation shows significant information loss if shared effects between different views or cohorts are not considered.

The computation time for the above methods on a standard computer is recorded in the parenthesis after each method name. For methods without a detailed time, they take less than 10 minutes. The more modules included in the model, the more computation time is required.

\begin{table}[H]
\setlength{\tabcolsep}{2pt}
\caption{Mean imputation RSE over 10 replications on the TCGA data, with
about 14.5\% of entries missing, by method (rows) and class of missing entry
(columns). Computation time on a standard computer is in parentheses;
methods without a time ran in under 10 minutes. Bold entries mark the
smallest value in each column.}
\label{paper3_tab:miss_real_data}
\begin{tabular}{lccccccc}
\hline
Method & entry & full col & full col & partial col & partial col & row & all missing \\
(computation time) & only & only & and row & only & and row & only & together \\ \hline
baRRR (8.6 hrs)            & 0.352 & 0.740 & 0.854 & 0.523 & 0.753 & 0.608 & $\boldsymbol{0.557}$ \\
BIDIFAC+ (2.5 hrs)         & $\boldsymbol{0.351}$ & 1.000 & 1.000 & 0.574 & 0.827 & 0.662 & 0.648 \\
bRRR (6.0 hrs)             & 0.600 & 0.614 & 0.888 & 0.620 & 0.892 & 0.889 & 0.680 \\
baRRR w/o cohort (7.6 hrs) & 0.353 & 0.900 & 0.961 & 0.554 & 0.811 & 0.650 & 0.615 \\
maRRR (8.0 hrs)            & 0.357 & 0.905 & 0.967 & $\boldsymbol{0.467}$ & 0.786 & 0.679 & 0.603 \\
aRRR                       & 0.428 & 0.728 & $\boldsymbol{0.836}$ & 0.531 & $\boldsymbol{0.701}$ & $\boldsymbol{0.596}$ & 0.572 \\
aRRR, 27 ind               & 0.400 & 0.741 & 0.844 & 0.473 & 0.860 & 0.852 & 0.614 \\
NNapp, 27 ind              & 0.401 & 1.000 & 1.000 & 0.479 & 1.000 & 1.000 & 0.717 \\
NNapp, 27 ind w/o cohort   & 0.731 & 1.001 & 1.003 & 0.813 & 1.004 & 1.005 & 0.886 \\
NNreg, 27 ind              & 0.572 & $\boldsymbol{0.607}$ & 0.887 & 0.589 & 0.896 & 0.897 & 0.665 \\
NNreg, 27 ind w/o cohort   & 0.687 & 0.749 & 1.002 & 0.694 & 1.027 & 1.040 & 0.791 \\ \hline
\end{tabular}
\end{table}

\section{Discussion}

We have proposed baRRR, which integrates bidimensionally linked data by decomposing variation into covariate-driven and auxiliary low-rank modules, each of which may be shared across any subset of sources and cohorts. The framework contains maRRR and BIDIFAC+ as special cases, and its algorithms impute entrywise and blockwise missing data as part of the fit. In simulations, baRRR recovered both shared and individual-level signal and had the lowest imputation error for every class of missing entry, including samples with all outcomes missing, where imputation amounts to predicting outcomes from covariates. In the TCGA application, auxiliary structure accounted for most of the variation, while mutation-driven variation was concentrated in modules spanning small groups of cancer types, such as the module shared by ovarian and uterine cancers, rather than in pan-cancer structure.

Several limitations point to future work. First, the module structure must be specified in advance,  so it would be hard for larger views or cohorts since there can be a lot of combinations; we took it from BIDIFAC+ applied to the same data, so module-level interpretations are conditional on that selection, and a procedure that selects modules jointly with the covariate effects would be preferable. Second, our penalty weights are fixed by a random matrix argument and an empirical Bayes approach in the spirit of \cite{lock2024empirical} could instead estimate module-specific shrinkage from the data. Third, on the TCGA data the gain in imputation accuracy over simpler special cases such as aRRR \citep{wang2023multiple} was small relative to the added computation, so baRRR's main advantage there lies in its interpretable decomposition; faster algorithms, for example based on randomized singular value decompositions, would help in larger studies. Fourth, our imputation analyses assume values are missing completely at random, which may not hold when entire platforms are missing for some samples.

Other directions include formal inference on which mutations drive each module, statistical error bounds and convergence rates for the estimator, and nonlinear extensions. baRRR models both covariate effects and auxiliary structure as linear and low rank; deep learning approaches to multi-view integration \citep{andrew2013deep,Wang2021DeepIDA,wang2024interpretable} capture nonlinear relationships, in some cases while preserving interpretability through feature ranking, and incorporating such components into the module framework is a promising extension. Omics data are also increasingly collected over additional dimensions such as time or tissue. Extending baRRR to linked tensors would build on tensor decomposition methods \citep{kolda2009tensor} and their multi-source extensions \citep{kim2023bayesian,kang2025multiple}, and treating repeated measurements over time as functional data \citep{ramsay2005functional} would draw on recent work on functional discriminant analysis \citep{liu2025multivariate}.


\section*{Acknowledgements}

This work was supported by the National Institutes of Health (NIH) grants R01-GM130622 and R35-GM161467. 

\section*{Availability}

Code to perform the baRRR method is available at \url{https://github.com/JiuzhouW/baRRR}. \newpage

\def\spacingset#1{\renewcommand{\baselinestretch}%
{#1}\small\normalsize} \spacingset{0.5}

\bibliographystyle{Chicago}
\bibliography{Bibliography}

\newpage

\section*{Appendix}

\begin{appendices}

\section{Optimization Algorithms}
\label{paper3_opt_alg}
\begin{algorithm}
    \caption{Alternating least squares with $L_2$ penalties for objective~\eqref{paper3_ObjUV}} \label{ALSUV_barrr}
  \begin{algorithmic}[1]
    \INPUT Covariates $\bYdd$ and corresponding multivariate outcomes $\bXdd$; penalizing terms $\lambda_{Bh}, \lambda_S$; binary indicator matrices $\bR_{Bh},\bC_{Yh}, \bC_S,\bR_S$; rank upper bounds $r_{B,upper}$ and $r_{S,upper}$
    \OUTPUT $\bB,\bS$
    \State \textbf{Initialization} Construct $\{\bYhdk\}$ based on $\bC_{Yh}$. Assign initialized numbers for each entry of $\{\bU_{Bh}^{(k)},\bV_{Bh}^{(k)}\}^K_{k=1}, \{\bU_S^{(l)},\bV_S^{(l)}\}^L_{l=1}$ with the number of columns/rows determined by $r_{B,upper}$ and $r_{S,upper}$
    \While{convergence criterion does not meet}
    \For{$h=1,...,H$}
    \For{$k=1,...,K$}
      \State Compute residual matrix $\bXdd^{(hk)} = \bXdd - \sum^H_{h'=1,h'\neq h} \sum^{K_h}_{k'=1,k'\neq k} \bU_{Bh'}^{(k')}\bV_{Bh'}^{(k')T} \mathbf{Y}_{h'\boldsymbol{\cdot}}^{(k')} - \sum^{L}_{l=1}  \bU_S^{(l)}\bV_S^{(l)T}$
      \State Set $\bX_{ij}^{(hk)} = \mathbf{0}$ where $\bC_{Yh}[j,l]=0$ for $j=1,...,J$ or $\bR_{Bh}[i,l]=0$ for $i=1,...,I$
      \State Update $\bU_{Bh}^{(k)} = \bXdd^{(hk)}\bYhdkT\bV_{Bh}^{(k)}(\bV_{Bh}^{(k)T}\bYhdk\bYhdkT\bV_{Bh}^{(k)} + \lambda_B^{(k)}\bI_{r_B})^{-1}$
      \State Update $vec(\bV_{Bh}^{(k)}) = [(\bU_{Bh}^{(k)T}\bU_{Bh}^{(k)}) \bigotimes (\bYhdk\bYhdkT) + \lambda_B^{(k)}\bI_{q * r_B} ]^{-1} vec[\bYhdk(\bXdd^{(hk)T})\bU_{Bh}^{(k)}]$
      \State Transform $vec(\bV_{Bh}^{(k)})$ to $\bV_{Bh}^{(k)}$
      \EndFor
      \EndFor
     \For{$l=1,..,L$}
      \State Compute the residual matrix $\bXddl = \bXdd - \sum^H_{h=1} \sum^{K_h}_{k=1} \bU_{Bh}^{(k)}\bV_{Bh}^{(k)T} \bYhdk - \sum^{L}_{l'=1,l'\neq l}  \bU_S^{(l')}\bV_S^{(l')T}$
      \State Set $\bX_{ij}^{(l)} = \mathbf{0}$ where $\bC_S[j,l]=0$ for $j=1,...,J$ or $\bR_S[i,l]=0$ for $i=1,...,I$
    \State Update $\bU_S^{(l)} = \bXddl\bV_S^{(l)}(\bV_S^{(l)T}\bV_S^{(l)} + \lambda_S^{(l)}\bI_{r_S})^{-1}$
      \State Update $\bV_S^{(l)} = \bXdd^{(l)T}\bU_S^{(l)}(\bU_S^{(l)T}\bU_S^{(l)} + \lambda_S^{(l)}\bI_{r_S})^{-1}$
     \EndFor
    \EndWhile
     \State Set $\bBdhk=\bU_{Bh}^{(k)}\bV_{Bh}^{(k)T}$ for all $k=1,..,K$,$h=1,...,H$, and $\bSddl=\bU_S^{(l)}\bV_S^{(l)T}$ for all $l=1,..,L$
  \end{algorithmic}
\end{algorithm}

Algorithm \ref{paper3_opt_alg} describes an alternating least squares algorithm to solve the objective with $L_2$ penalties.  The symbol $\bigotimes$ represents the Kronecker product. Note that \cref{ALSUV_barrr} does not require the columns of $\bYd^{(k)}$ to be orthogonal. When $\bYd^{(k)}$ is semi-orthogonal, we develop an alternative approach by iterative soft-singular value thresholding estimators for \eqref{paper3_ObjZ} in \cref{ALSBS_barrr}.

\begin{algorithm}
    \caption{Iterative soft-thresholding matrix decompositions for objective~\eqref{paper3_ObjZ}} \label{ALSBS_barrr}
  \begin{algorithmic}[1]
    \INPUT Orthogonal covariates $\bYdd$ and corresponding multivariate outcomes $\bXdd$; penalizing terms $\lambda_{Bh}, \lambda_S$; binary indicator matrices $\bR_{Bh},\bC_{Yh}, \bC_S,\bR_S$
    \OUTPUT $\bB,\bS$
    \State \textbf{Initialization} Construct $\{\bYhdk\}$ based on $\bC_{Yh}$. Assign initialized numbers for each entry of $\{\bBdhk\},\{\bSddl\}$
    \While{convergence criterion does not meet}
    \For{$h=1,..,H$}
    \For{$k=1,..,K_h$}
    \State Compute the residual matrix $\bXdd^{(hk)} = \bXdd - \sum^{H}_{h'=1,h'\neq h} \sum^{K_h}_{k'=1,k'\neq k} \mathbf{B}_{\boldsymbol{\cdot}h'}^{(k')}\mathbf{Y}_{h'\boldsymbol{\cdot}}^{(k')} - \sum^{L}_{l=1}  \bSddl$
    \State Set $\bX_{ij}^{(hk)} = \mathbf{0}$ where $\bC_{Yh}[j,l]=0$ for $j=1,...,J$ or $\bR_{Bh}[i,l]=0$ for $i=1,...,I$
    \State Compute the SVD of $\bXdd^{(hk)}\bYhdk$, i.e. $\bXdd^{(hk)}\bYhdkT = \bL_B^{(kh)}\bD_B^{(kh)}\bR_{B}^{(kh)}$
    \State Update $\bBdhk = \bL_{B}^{(kh)}\widehat{\bD}_B^{(kh)}\bR_{B}^{(kh)}$ where $\widehat{\bD}_B^{(kh)}$ is a diagonal matrix with $\widehat{\bD}_B^{(kh)}[r,r] = max(\bD_B^{(kh)}[r,r]-\lambda_{Bh}^{(k)},0)$ for $r=1,2,...$ on its diagonal entries and zero otherwise
    \EndFor
    \EndFor
    \For{$l=1,..,L$}
    \State Compute the residual matrix $\bXddl = \bXdd - \sum^{K}_{k=1} \bBdhk \bYhdk - \sum^{L}_{l=1,l'\neq l}  \bSddlt$
      \State Set $\bX_{ij}^{(l)} = \mathbf{0}$ where $\bC_S[j,l]=0$ for $j=1,...,J$ or $\bR_S[i,l]=0$ for $i=1,...,I$
    \State Compute the SVD of $\bXddl$, i.e. $\bXddl = \bL_{S}^{(l)}\bD_S^{(l)}\bR_{S}^{(l)}$
    \State Update $\bSddl = \bL_{S}^{(l)}\widehat{\bD}_S^{(l)}\bR_{S}^{(l)}$ where $\widehat{\bD}_S^{(l)}$ is a diagonal matrix with $\widehat{\bD}_S^{(l)}[r,r] = max(\bD_S^{(l)}[r,r]-\lambda_S^{(l)},0)$ for $r=1,2,...$ on its diagonal entries and zero otherwise
    \EndFor  
    \EndWhile
  \end{algorithmic}
\end{algorithm}

\section{Proofs}
\subsection{Proof of Theorem 2}
\label{paper3_proof_thm2}
\begin{proof}
We demonstrate the proof for \cref{paper3_ObjZ} and the proof for \cref{paper3_ObjUV} will be identical. For simplicity, we provide proof while $H = 1$ since a similar proof can be made for a larger $H$.
Define our optimization problem in \ref{paper3_ObjZ} as follows:
\begin{align*}
&\min_{\{\bBdk\}^K_{k=1},\{\bS_{\boldsymbol{\cdot}\boldsymbol{\cdot}}^{(l)}\}^L_{l=1}} F(\bBd^{(1)},...,\bBdk,\bSddONE,...,\bSddL)\\
=&\min_{\{\bBdk\}^K_{k=1},\{\bS_{\boldsymbol{\cdot}\boldsymbol{\cdot}}^{(l)}\}^L_{l=1}} f(\bBd^{(1)},...,\bBdk,\bSddONE,...,\bSddL) + \sum_{i=1}^{K+L}r_i(\bBd^{(1)},...,\bBdk,\bSddONE,...,\bSddL)\\
    =&\min_{\{\bBdk\}^K_{k=1},\{\bS_{\boldsymbol{\cdot}\boldsymbol{\cdot}}^{(l)}\}^L_{l=1}} \{ \frac{1}{2} ||\bX_{\boldsymbol{\cdot}\boldsymbol{\cdot}} - \sum^{K}_{k=1} \bBdk \bYd^{(k)} - \sum^{L}_{l=1}  \bSddl||_F^2 + \sum^{K}_{k=1} \lambda_B^{(k)} ||\bBdk||_{*} +\sum^{L}_{l=1} \lambda_S^{(l)} ||\bS_{\boldsymbol{\cdot}\boldsymbol{\cdot}}^{(l)}||_{*}  \}
\end{align*}
We have $K+L$ blocks in total. Let $\bx_i^t$ denote the value of $\bx_i$ after its $t$th update, and let \\ $f_i^t(\bx) \triangleq f(\bBd^{(1)t},...,\bBd^{(i-1)t},\bBd^{(i)},\bBd^{(i+1)t-1},...,\bBd^{(K)t-1},\bS_{\boldsymbol{\cdot\cdot}}^{(1)t-1},...,\bS_{\boldsymbol{\cdot\cdot}}^{(L)t})$ for $i = 1,...,K, i.e.$ blocks in $\bBdk, k = 1,...,K,$ \\and $f_i^t(\bx) \triangleq f(\bBd^{(1)t},...,\bBd^{(K)t},\bS_{\boldsymbol{\cdot\cdot}}^{(1)t},...,\bS_{\boldsymbol{\cdot\cdot}}^{(i-1)t},\bS_{\boldsymbol{\cdot\cdot}}^{(i)},\bS_{\boldsymbol{\cdot\cdot}}^{(i+1)t-1},...,\bS_{\boldsymbol{\cdot\cdot}}^{(L)t-1})$ for $i = K+1,...,K+L, i.e.$ blocks in $\bSddl, l = 1,...,L.$ 
Follow the definition in \citep{xu2013block}, the loss function is block multi-convex, since for each $\bBdk$ and $\bSddl$, it is a convex function of $\bBdk$ or $\bSddl$ while all the other blocks are fixed. Our two optimization algorithms are both alternating least square which gives the optimal solution in each iteration. Once the assumptions in Lemma 2.6 in \citep{xu2013block} are verified, we can claim our block coordinate descent method outputs global minima by Theorem 2.8.

Now we verify the following five conditions:
\begin{enumerate}
    \item (General Assumption 1) F is continuous in its domain $dom(F) = \bR^{K \times p \times q + L \times p \times n}$ and $\inf_{x \in dom(F)} F(x) > -\infty$. The loss function has a Nash point\\ $\bar{\bx} = (\bar{\bB}_{\boldsymbol{\cdot\cdot}}^{(1)},...,\bar{\bB}_{\boldsymbol{\cdot\cdot}}^{(k)},...,\bar{\bB}_{\boldsymbol{\cdot\cdot}}^{(K)},\bSddONEbar,...,\bSddlbar,...,\bSddLbar)$, i.e. a block coordinate minimizer. For 
    each coordinate $k = 1,...,K$, $F(\bar{\bB}_{\boldsymbol{\cdot\cdot}}^{(1)},...,\bar{\bB}_{\boldsymbol{\cdot\cdot}}^{(k)},...,\bar{\bB}_{\boldsymbol{\cdot\cdot}}^{(K)},\bSddONEbar,...,\bSddLbar) \leq F(\bar{\bB}_{\boldsymbol{\cdot\cdot}}^{(1)},...,\bBdk,...,\bar{\bB}_{\boldsymbol{\cdot\cdot}}^{(K)},\bSddONEbar,...,\bSddLbar)$ and for $l = 1,...,L, $ \\$F(\bar{\bB}_{\boldsymbol{\cdot\cdot}}^{(1)},...,\bar{\bB}_{\boldsymbol{\cdot\cdot}}^{(K)},\bSddONEbar,...,\bSddlbar,...,\bSddLbar) \leq F(\bar{\bB}_{\boldsymbol{\cdot\cdot}}^{(1)},...,\bar{\bB}_{\boldsymbol{\cdot\cdot}}^{(K)},\bSddONEbar,...,\bSddl,...,\bSddLbar)$.
  \item (General Assumption 2)  Each block $i (i=1,...,K+L)$ is updated by the same scheme among block coordinate descent (BCD) method of Gauss–Seidel type, which minimizes $F$ cyclically over each of $\bBd^{(1)},...,\bBdk,\bSddONE,...,\bSddL$ while fixing the remaining blocks at their last updated values, for all updates. In addition, there exist constants $0< l_i \leq L_i <\infty, i = 1,...,K+L$, such that $f_i^k$ is strongly convex with modulus $l_i \leq L^{t-1}_i \leq L_i$, where $t$ denotes the $t$th update for one block, i.e. $f_i^t$ is $L^{t-1}_i$ strongly convex. For block $i=1,...,K$, $f_i^t = \frac{1}{2} ||\bX_{\boldsymbol{\cdot}\boldsymbol{\cdot}} - \sum^{i-1}_{s=1} \bBd^{(s)t} \bY_{\boldsymbol{\cdot}}^{(s)} -\bBd^{(i)} \bY_{\boldsymbol{\cdot}}^{(i)} - \sum^{K}_{s=i+1} \bBd^{(s)t-1} \bY_{\boldsymbol{\cdot}}^{(s)} - \sum^{L}_{l=1}  \bSddl^{t-1}||_F^2 $. The gradient of $f_i^t$ is $\nabla f_i^t = \bB_i \bY_{\boldsymbol{\cdot}}^{(i)} \bY_{\boldsymbol{\cdot}}^{(i)T} - (\bX_{\boldsymbol{\cdot}} - \sum^{i-1}_{s=1} \bBd^{(s)t} \bY_{\boldsymbol{\cdot}}^{(s)} - \sum^{K}_{s=i+1} \bB^{(s)t-1} \bY_{\boldsymbol{\cdot}}^{(s)} - \sum^{L}_{l=1}  \bSddl^{t-1}) \bY_{\boldsymbol{\cdot}}^{(i)T}$ and the Hessian matrix of $f_i^t$ is $\boldsymbol{H}f_i^t = \bY_{\boldsymbol{\cdot}}^{(i)} \bY_{\boldsymbol{\cdot}}^{(i)T}$. The Hessian is constant at each update step. For all K blocks, to qualify the strong convexity requirement, we need all $\boldsymbol{H}_f = \bY_{\boldsymbol{\cdot}}^{(i)} \bY_{\boldsymbol{\cdot}}^{(i)T} \geq L_i \bI$, which means the eigenvalues of $\bY_{\boldsymbol{\cdot}}^{(i)} \bY_{\boldsymbol{\cdot}}^{(i)T} - L_i \bI$ are all non-negative. This is based on specific values for $\bY_{\boldsymbol{\cdot}}^{(i)}$ but it is plausible. For block $ i = K+1,...,K+L, f_i^t = \frac{1}{2} ||\bXdd - \sum^{K}_{k=1} \bBdk\bYd^{(k)}  - \sum^{i-K-1}_{s=1}  \bSdds^{t} - \bSddik - \sum^{K}_{s=i-K}  \bSdds^{t-1}||_F^2 $. The Hessian is identity matrix for all blocks so it is easy to set $L_i = 1$ to qualify the strong convexity.
 
    \item (Lemma 2.6 assumption 1) The gradient of the squared Frobenius norm is linear so $f$ is Lipschitz continuous on any bounded set. 
    \item (Lemma 2.6 assumption 2)  Since Frobenious norm $f$ is strongly convex as we show in Condition 2 and nuclear norm $r$ is convex as well, $F = f + r$ is also strongly convex. And locally strongly convex functions satisfy the KL-inequality (in the \cite{xu2013block} (2.14)). Therefore, we claim that $F$ satisfies the Kurdyka–Lojasiewicz inequality at $\bar{\bx}$ ($\bar{\bx}$ is a limit point of sequence $\{\bx^k\}$.
    \item (Lemma 2.6 assumption 3) The starting point $\bx^0$ is sufficiently close to $\bar{\bx}$, and the $t$th update $F_t>\bar{F}, \forall t \geq 0$ since the blockwise gradient descent is descending at each step. 
\end{enumerate}
\end{proof}
\section{More Details On The Simulations}
\label{paper3_more_simu}

We also tested the imputation accuracy of baRRR under various signal levels and missing proportion, when there are more signals in the covaraite effects, i.e. true rank of $\bBdk$ is 3 other than 1 in the main context.

\begin{table}[H]
\label{paper3_tab:miss_rank3}
\caption{Imputation relative squared error(RSE) for baRRR under various signal levels and missing proportions (true rank of $\bBdk$ is 3). Each RSE is the mean of 100 replications.}
\centering
\setlength{\tabcolsep}{3pt} 
\begin{tabular}{lcccccccc}
\hline
Signal        & \multicolumn{1}{c}{Missing}     & \multicolumn{1}{c}{entry} &  \multicolumn{1}{c}{full col} & \multicolumn{1}{c}{full col} & \multicolumn{1}{c}{partial col} & \multicolumn{1}{c}{partial col} & 
\multicolumn{1}{c}{row}  &\multicolumn{1}{c}{all missing} \\
levels        & \multicolumn{1}{c}{rate} & \multicolumn{1}{c}{only}  &  \multicolumn{1}{c}{only}        & \multicolumn{1}{c}{and row}     & \multicolumn{1}{c}{only}     & \multicolumn{1}{c}{and row}  & 
\multicolumn{1}{c}{only} &\multicolumn{1}{c}{together}    \\ \hline
standard   & 0.04                            & 0.141                     & 0.573                    & 0.844                           & 0.415                        & 0.786                           & 0.675                        & 0.463                           \\
           & 0.08                            & 0.149                     & 0.582                    & 0.815                           & 0.421                        & 0.784                           & 0.678                        & 0.457                           \\
           & 0.2                             & 0.182                     & 0.597                    & 0.831                           & 0.452                        & 0.779                           & 0.697                        & 0.485                           \\
           & 0.4                             & 0.262                     & 0.641                    & 0.851                           & 0.524                        & 0.818                           & 0.747                        & 0.546                           \\
           & 0.8                             & 0.596                     & 0.827                    & 0.940                           & 0.782                        & 0.931                           & 0.901                        & 0.777                           \\ \hline
high BY    & 0.04                            & 0.050                     & 0.104                    & 0.652                           & 0.085                        & 0.585                           & 0.576                        & 0.223                           \\
           & 0.08                            & 0.053                     & 0.109                    & 0.580                           & 0.090                        & 0.600                            & 0.572                        & 0.207                           \\
           & 0.2                             & 0.063                     & 0.127                    & 0.591                           & 0.106                        & 0.587                           & 0.580                        & 0.224                           \\
           & 0.4                             & 0.086                     & 0.175                    & 0.610                           & 0.150                        & 0.616                           & 0.602                        & 0.259                           \\
           & 0.8                             & 0.217                     & 0.412                    & 0.742                           & 0.367                        & 0.735                           & 0.712                        & 0.434                           \\ \hline
high S     & 0.04                            & 0.135                     & 0.872                    & 0.952                           & 0.259                        & 0.487                           & 0.355                        & 0.399                           \\
global     & 0.08                            & 0.142                     & 0.878                    & 0.956                           & 0.261                        & 0.439                           & 0.357                        & 0.408                           \\
           & 0.2                             & 0.165                     & 0.890                    & 0.958                           & 0.289                        & 0.465                           & 0.376                        & 0.432                           \\
           & 0.4                             & 0.219                     & 0.913                    & 0.968                           & 0.354                        & 0.534                           & 0.437                        & 0.489                           \\
           & 0.8                             & 0.439                     & 0.979                    & 0.994                           & 0.641                        & 0.790                           & 0.676                        & 0.694                           \\ \hline
high S  & 0.04                            & 0.144                     & 0.875                    & 0.947                           & 0.282                        & 0.946                           & 0.926                        & 0.581                           \\
column           & 0.08                            & 0.150                     & 0.882                    & 0.955                           & 0.290                        & 0.949                           & 0.928                        & 0.562                           \\
           & 0.2                             & 0.175                     & 0.892                    & 0.960                           & 0.321                        & 0.952                           & 0.936                        & 0.583                           \\
           & 0.4                             & 0.233                     & 0.914                    & 0.969                           & 0.397                        & 0.965                           & 0.954                        & 0.624                           \\
           & 0.8                             & 0.464                     & 0.979                    & 0.994                           & 0.708                        & 0.994                           & 0.991                        & 0.779                           \\ \hline
high S  & 0.04                            & 0.143                     & 0.875                    & 0.950                           & 0.847                        & 0.944                           & 0.381                        & 0.550                           \\
  row         & 0.08                            & 0.15                      & 0.886                    & 0.951                           & 0.851                        & 0.957                           & 0.394                        & 0.573                           \\
           & 0.2                             & 0.176                     & 0.895                    & 0.960                           & 0.869                        & 0.953                           & 0.418                        & 0.594                           \\
           & 0.4                             & 0.232                     & 0.915                    & 0.969                           & 0.899                        & 0.965                           & 0.489                        & 0.644                           \\
           & 0.8                             & 0.459                     & 0.979                    & 0.994                           & 0.975                        & 0.993                           & 0.727                        & 0.793                           \\ \hline
high S  & 0.04                            & 0.155                     & 0.875                    & 0.973                           & 0.852                        & 0.956                           & 0.930                        & 0.716                           \\
 individual          & 0.08                            & 0.161                     & 0.885                    & 0.949                           & 0.856                        & 0.960                           & 0.930                        & 0.705                           \\
           & 0.2                             & 0.188                     & 0.894                    & 0.962                           & 0.871                        & 0.952                           & 0.938                        & 0.721                           \\
           & 0.4                             & 0.247                     & 0.914                    & 0.970                           & 0.901                        & 0.965                           & 0.955                        & 0.746                           \\
           & 0.8                             & 0.479                     & 0.979                    & 0.994                           & 0.976                        & 0.993                           & 0.991                        & 0.842                        \\  \hline
\end{tabular}
\end{table}

\section{More Details On The Real Data Analysis}
\label{paper3_append_caner_names}

\begin{table}[H]
\caption{Names for cancer types, including the cancer study abbreviations, sample sizes and study names, sourcing from \href{https://gdc.cancer.gov/resources-tcga-users/tcga-code-tables/tcga-study-abbreviations}{National Cancer Institute}.}
\centering
\begin{tabular}{lll}
\hline
Abbreviation & Sample Size & Study Name                                                       \\\hline
ACC          & 77          & Adrenocortical carcinoma                                         \\
BLCA         & 129         & Bladder Urothelial Carcinoma                                     \\
BRCA         & 976         & Breast invasive carcinoma                                        \\
CESC         & 193         & Cervical squamous cell carcinoma  \\
         &          & and endocervical adenocarcinoma \\
COAD         & 147         & Colon adenocarcinoma                                             \\
ESCA         & 184         & Esophageal carcinoma                                             \\
HNSC         & 279         & Head and Neck squamous cell carcinoma                            \\
KICH         & 66          & Kidney Chromophobe                                               \\
KIRC         & 415         & Kidney renal clear cell carcinoma                                \\
KIRP         & 161         & Kidney renal papillary cell carcinoma                            \\
LGG          & 283         & Brain Lower Grade Glioma                                         \\
LIHC         & 195         & Liver hepatocellular carcinoma                                   \\
LUAD         & 230         & Lung adenocarcinoma                                              \\
LUSC         & 178         & Lung squamous cell carcinoma                                     \\
OV           & 115         & Ovarian serous cystadenocarcinoma                                \\
PAAD         & 150         & Pancreatic adenocarcinoma                                        \\
PCPG         & 179         & Pheochromocytoma and Paraganglioma                               \\
PRAD         & 331         & Prostate adenocarcinoma                                          \\
READ         & 64          & Rectum adenocarcinoma                                            \\
SARC         & 245         & Sarcoma                                                          \\
SKCM         & 342         & Skin Cutaneous Melanoma                                          \\
STAD         & 275         & Stomach adenocarcinoma                                           \\
TGCT         & 133         & Testicular Germ Cell Tumors                                      \\
THCA         & 400         & Thyroid carcinoma                                                \\
THYM         & 119         & Thymoma                                                          \\
UCEC         & 242         & Uterine Corpus Endometrial Carcinoma                             \\
UCS          & 57          & Uterine Carcinosarcoma                                           \\
 \hline                                  
\end{tabular}
\label{paper3_tab: cancer_mutation}
\end{table}

\begin{table}[H]
\caption{Names for somatic mutations. Gene labels for 50 somatic mutation data in order.}
\centering
\begin{tabular}{cccccccc}
\hline
Index & Mutation & Index & Mutation & Index & Mutation & Index & Mutation \\\hline
1     & TP53     & 14    & APC      & 27    & FRG1B    & 40    & DST      \\
2     & TTN      & 15    & DNAH5    & 28    & AHNAK2   & 41    & MLL3     \\
3     & MUC16    & 16    & OBSCN    & 29    & APOB     & 42    & MLL2     \\
4     & PIK3CA   & 17    & MUC17    & 30    & ARID1A   & 43    & MACF1    \\
5     & CSMD3    & 18    & FAT4     & 31    & LRP2     & 44    & DNAH9    \\
6     & LRP1B    & 19    & HMCN1    & 32    & XIRP2    & 45    & BRAF     \\
7     & KRAS     & 20    & CSMD1    & 33    & Unknown  & 46    & DNAH11   \\
8     & RYR2     & 21    & MUC5B    & 34    & DMD      & 47    & DNAH8    \\
9     & MUC4     & 22    & ZFHX4    & 35    & PKHD1L1  & 48    & CSMD2    \\
10    & FLG      & 23    & FAT3     & 36    & RYR1     & 49    & MUC2     \\
11    & SYNE1    & 24    & SPTA1    & 37    & RYR3     & 50    & ABCA13   \\
12    & USH2A    & 25    & GPR98    & 38    & NEB      &       &          \\
13    & PCLO     & 26    & PTEN     & 39    & PCDH15   &       &          \\\hline
\end{tabular}
\label{paper3_tab: gene_label}
\end{table}

\end{appendices}


\spacingset{1.5}

\allowdisplaybreaks

\newenvironment{definition}[1][Definition]{\begin{trivlist}
		\item[\hskip \labelsep {\bfseries #1}]}{\end{trivlist}}
\newenvironment{example}[1][Example]{\begin{trivlist}
		\item[\hskip \labelsep {\bfseries #1}]}{\end{trivlist}}
\newenvironment{rmq}[1][Remark]{\begin{trivlist}
		\item[\hskip \labelsep {\bfseries #1}]}{\end{trivlist}}

%
%


\bigskip









\end{document}